\documentclass{IEEEtran}
\usepackage{cite}
\usepackage[pdftex]{graphicx}
\usepackage[cmex10]{amsmath}
\usepackage{amssymb}
\usepackage{amsthm}
\usepackage{algorithm}
\usepackage{algorithmic}
\usepackage[caption=false,font=footnotesize]{subfig}
\usepackage{url}
\usepackage{multirow}

\DeclareMathOperator{\real}{Re}
\DeclareMathOperator{\imag}{Im}
\DeclareMathOperator{\vect}{vec}

\newtheorem{thm}{Theorem}
\newtheorem{prop}[thm]{Proposition}
\newtheorem{cor}[thm]{Corollary}
\newtheorem{lem}[thm]{Lemma}

\title{Polar $n$-Complex and $n$-Bicomplex Singular Value Decomposition and Principal Component Pursuit}

\author{Tak-Shing~T.~Chan,~\IEEEmembership{Member,~IEEE} and Yi-Hsuan~Yang,~\IEEEmembership{Member,~IEEE}%
\thanks{Manuscript received August 26, 2015; revised May 26, 2016 and July 16, 2016; accepted September 3, 2016. Date of publication Month xx, 2016; date of current version September 4, 2016. This work was supported by a grant from the Ministry of Science and Technology under the contract MOST102-2221-E-001-004-MY3 and the Academia Sinica Career Development Program. The associate editor coordinating the review of this manuscript and approving it for publication was Prof.~Masahiro Yukawa.}%
\thanks{The authors are with the Research Center for Information Technology Innovation, Academia Sinica, Taipei 11564, Taiwan (e-mail: takshingchan@citi.sinica.edu.tw; yang@citi.sinica.edu.tw).}%
\thanks{Digital Object Identifier 10.1109/TSP.2016.2612171}}

\IEEEpubid{\copyright~2016 IEEE. Personal use is permitted. For any other purposes, permission must be obtained from the IEEE by emailing pubs-permissions@ieee.org.}

\begin{document}

\maketitle

\begin{abstract}
Informed by recent work on tensor singular value decomposition and circulant algebra matrices, this paper presents a new theoretical bridge that unifies the hypercomplex and tensor-based approaches to singular value decomposition and robust principal component analysis. We begin our work by extending the principal component pursuit to Olariu's polar $n$-complex numbers as well as their bicomplex counterparts. In so doing, we have derived the polar $n$-complex and $n$-bicomplex proximity operators for both the $\ell_1$- and trace-norm regularizers, which can be used by proximal optimization methods such as the alternating direction method of multipliers. Experimental results on two sets of audio data show that our algebraically-informed formulation outperforms tensor robust principal component analysis. We conclude with the message that an informed definition of the trace norm can bridge the gap between the hypercomplex and tensor-based approaches. Our approach can be seen as a general methodology for generating other principal component pursuit algorithms with proper algebraic structures.
\end{abstract}


\section{Introduction}
\label{sec:intro}

\IEEEPARstart{T}{he} robust principal component analysis (RPCA) \cite{Candes11} has received a lot of attention lately in many application areas of signal processing \cite{Huang12,Ikemiya15,Peng10,Bouwmans14}. The ideal form of RPCA decomposes the input $\mathbf{X}\in\mathbb{R}^{l \times m}$ into a low-rank matrix $\mathbf{L}$ and a sparse matrix $\mathbf{S}$:
\begin{equation}
\label{eq:1}
\min_{\mathbf{L},\mathbf{S}}\mathrm{rank}(\mathbf{L})+\lambda\|\mathbf{S}\|_0\mbox{\quad s.t.\quad}\mathbf{X}=\mathbf{L}+\mathbf{S},
\end{equation}
where $\|\cdot\|_0$ returns the number of nonzero matrix elements. Owing to the NP-hardness of the above formulation, the principal component pursuit (PCP) \cite{Candes11} has been proposed to solve this relaxed problem instead \cite{Lin09}:
\begin{equation}
\label{eq:2}
\min_{\mathbf{L},\mathbf{S}}\|\mathbf{L}\|_*+\lambda\|\mathbf{S}\|_1\mbox{\quad s.t.\quad}\mathbf{X}=\mathbf{L}+\mathbf{S}\,,
\end{equation}
where $\|\cdot\|_*$ is the trace norm (sum of the singular values), $\|\cdot\|_1$ is the entrywise $\ell_1$-norm, and $\lambda$ can be set to $c/\sqrt{\max(l,m)}$ where $c$ is a positive parameter \cite{Candes11,Huang12}. The trace norm and the $\ell_1$-norm are the tightest convex relaxations of the rank and the $\ell_0$-norm, respectively. Under somewhat general conditions \cite{Candes11}, PCP with $c=1$ has a high probability of exact recovery, though $c$ can be tuned if the conditions are not met.

\IEEEpubidadjcol

Despite its success, one glaring omission from the original PCP is the lack of complex (and hypercomplex) formulations. In numerous signal processing domains, the input phase has a significant meaning. For example in parametric spatial audio, spectrograms have not only spectral phases but inter-channel phases as well. For that reason alone, we have recently extended the PCP to the complex and the quaternionic cases \cite{Chan16}. However, there exists inputs with dimensionality greater than four, such as microphone array data, surveillance video from multiple cameras, or electroencephalogram (EEG) signals, which exceed the capability of quaternions. These signals may instead be represented by $n$-dimensional hypercomplex numbers, defined as \cite{Kantor89}
\begin{equation}
a=a_0+a_1e_1+\cdots+a_{n-1}e_{n-1},
\end{equation}
where $a_0,\ldots,a_{n-1}\in\mathbb{R}$ and $e_1\ldots,e_{n-1}$ are the imaginary units. Products of imaginary units are defined by an arbitrary $(n-1)\times(n-1)$ multiplication table, and multiplication follows the distributive rule \cite{Kantor89}. If we impose the multiplication rules
\begin{equation}
e_ie_j=\begin{cases}
-e_je_i, & i\neq j,\\
-1,0,\mathrm{\,or\,}1, & i=j,
\end{cases}
\end{equation}
and extend the algebra to include all $2^{n-1}$ combinations of imaginary units (formally known as multivectors):
\begin{equation}
\begin{split}
a= & \ a_0\\
& +a_1e_1+a_2e_2+\ldots\\
& +a_{1,2}e_1e_2+a_{1,3}e_1e_3+\ldots\\
& +\ldots\\
& +a_{1,2,\ldots,n-1}e_1e_2\ldots e_{n-1},
\end{split}
\end{equation}
then we have a Clifford algebra \cite{Lounesto01}. For example, the real, complex, and quaternion algebras are all Clifford algebras. Yet previously, Alfsmann \cite{Alfsmann06} suggests two families of $2^N$-dimensional hypercomplex numbers suitable for signal processing and argued for their superiority over Clifford algebras. One family starts from the two-dimensional hyperbolic numbers and the other one starts from the four-dimensional tessarines,\footnote{Hyperbolic numbers are represented by $a_0+a_1j$ where $j^2=1$ and $a_0,a_1\in\mathbb{R}$ \cite{Alfsmann06}. Tessarines are almost identical except that $a_0,a_1\in\mathbb{C}$ \cite{Alfsmann06}.} with dimensionality doubling up from there. Although initially attractive, the $2^N$-dimensional restriction (which also affects Clifford algebras) seems a bit limiting. For instance, if we have 100 channels to process, we are forced to use 128 dimensions (wasting 28). On the other hand, tensors can have arbitrary dimensions, but traditionally they do not possess rich algebraic structures. Fortunately, recent work on the tensor singular value decomposition (SVD) \cite{Kilmer08}, which the authors call the t-SVD, has begun to impose more structures on tensors \cite{Braman10,Kilmer11,Kilmer13}. Furthermore, a tensor PCP formulation based on t-SVD has also been proposed lately \cite{Zhang14}. Most relevantly, Braman \cite{Braman10} has suggested to investigate the relationship between t-SVD and Olariu's \cite{Olariu02} $n$-complex numbers (for arbitrary $n$). This is exactly what we need, yet the actual work is not forthcoming. So we have decided to begin our investigation with Olariu's polar $n$-complex numbers. Of special note is Gleich's work on the circulant algebra \cite{Gleich13}, which is isomorphic to Olariu's polar $n$-complex numbers. This observation simplifies our current work significantly. Nevertheless, the existing tensor PCP \cite{Zhang14} employs an ad hoc tensor nuclear norm, which lacks algebraic validity. So, in this paper, we remedy this gap by formulating the first proper $n$-dimensional PCP algorithm using the polar $n$-complex algebra.

Our contributions in this paper are twofold. First, we have extended PCP to the polar $n$-complex algebra and the polar $n$-bicomplex algebra (defined in Section \ref{sec:n-bicomplex}), via: 1) properly exploiting the circulant isomorphism for the polar $n$-complex numbers; 2) extending the polar $n$-complex algebra to a new polar $n$-bicomplex algebra; and 3) deriving the proximal operators for both the polar $n$-complex and $n$-bicomplex matrices by leveraging the aforementioned isomorphism. Second, we have provided a novel hypercomplex framework for PCP where algebraic structures play a central role.

This paper is organized as follows. In Section \ref{sec:n-complex}, we review polar $n$-complex matrices and their properties. We extend this to the polar $n$-bicomplex case in Section \ref{sec:n-bicomplex}. This leads to the polar $n$-complex and $n$-bicomplex PCP in Section \ref{sec:pcp}. Experiments are conducted in Sections \ref{sec:numerical} and \ref{sec:exp} to justify our approach. We conclude by describing how our work provides a new direction for future work in Section \ref{sec:conc}.

\section{The Polar $n$-Complex Numbers}
\label{sec:n-complex}

In this section we introduce polar $n$-complex matrices and their isomorphisms. These will be required in Section \ref{sec:pcp} for the formulation of polar $n$-complex PCP. Please note that the value of $n$ here does not have to be a power of two.

\subsection{Background}
\label{subsec:n-background}

Olariu's \cite{Olariu02} polar $n$-complex numbers, which we denote by $\mathbb{K}_n$, are $n$-dimensional ($n\geq 2$) extensions of the complex algebra, defined as
\begin{equation}
\label{eq:3}
p=a_0e_0+a_1e_1+\cdots+a_{n-1}e_{n-1}\in\mathbb{K}_n,
\end{equation}
where $a_0,a_1,\ldots,a_{n-1}\in\mathbb{R}$. The first imaginary unit is defined to be $e_0=1$ whereas $e_1,\ldots,e_{n-1}$ are defined by the multiplication table \cite{Olariu02}
\begin{equation}
\label{eq:4}
e_ie_k=e_{(i+k)\bmod n}.
\end{equation}
We call $\real p=a_0$ the real part of $p$ and $\imag_i p=a_i$ the imaginary parts of $p$ for $i=0,1,\ldots,n-1$. We remark that our imaginary index starts with $0$, which includes the real part, to facilitate a shorter definition of equations such as \eqref{eq:31} and \eqref{eq:38}. Multiplication follows the usual associative and commutative rules \cite{Olariu02}. The inverse of $p$ is the number $p^{-1}$ such that $pp^{-1}=1$ \cite{Olariu02}. Olariu named it the polar $n$-complex algebra because it is motivated by the polar representation of a complex number \cite{Olariu02} where $a+jb\in\mathbb{C}$ is represented geometrically by its modulus $\sqrt{a^2+b^2}$ and polar angle $\arctan(b/a)$. Likewise, the polar $n$-complex number in \eqref{eq:3} can be represented by its modulus
\begin{equation}
\label{eq:5}
|p|=\sqrt{a_0^2+a_1^2+\cdots+a_{n-1}^2}
\end{equation}
together with $\lceil n/2\rceil-1$ azimuthal angles, $\lceil n/2\rceil-2$ planar angles, and one polar angle (two if $n$ is even), totaling $n-1$ angles \cite{Olariu02}. To calculate these angles, let $[A_0,A_1,\ldots,A_{n-1}]^T$ be the discrete Fourier transform (DFT) of $[a_0,a_1,\ldots,a_{n-1}]^T$, defined by
\begin{equation}
\label{eq:6}
\begin{bmatrix}
A_0\\
A_1\\
\vdots\\
A_{n-1}
\end{bmatrix}=\mathbf{F}_n\begin{bmatrix}
a_0\\a_1\\
\vdots\\
a_{n-1}
\end{bmatrix},
\end{equation}
where $\omega_n=e^{-j2\pi/n}$ is a principal $n$th root of unity and
\begin{equation}
\label{eq:7}
\mathbf{F}_n=\frac{1}{\sqrt{n}}\begin{bmatrix}
1 & 1 & \cdots & 1\\
1 & \omega_n & \cdots & \omega_n^{n-1}\\
\vdots & \vdots & \ddots & \vdots\\
1 & \omega_n^{n-1} & \cdots & \omega_n^{(n-1)(n-1)}
\end{bmatrix},
\end{equation}
which is unitary, i.e., $\mathbf{F}_n^*=\mathbf{F}_n^{-1}$. For $k=1,\ldots,\lceil n/2\rceil-1$, the azimuthal angles $\phi_k$ can be calculated from \cite{Olariu02}
\begin{equation}
\label{eq:8}
A_k=|A_k|e^{-j\phi_k},
\end{equation}where $0\leq\phi_k<2\pi$. Note that we have reversed the sign of the angles as Olariu was a physicist so his DFT is our inverse DFT. Furthermore, for $k=2,\ldots,\lceil n/2\rceil-1$, the planar angles $\psi_{k-1}$ are defined by \cite{Olariu02}

\begin{equation}
\label{eq:9}
\tan\psi_{k-1}=\frac{|A_1|}{|A_k|},
\end{equation}
where $0\leq\psi_k\leq\pi/2$. The polar angle $\theta_+$ is defined as \cite{Olariu02}
\begin{equation}
\label{eq:10}
\tan\theta_+=\frac{\sqrt{2}|A_1|}{A_0},
\end{equation}
where $0\leq\theta_+\leq\pi$. Finally, for even $n$, there is an additional polar angle \cite{Olariu02},
\begin{equation}
\label{eq:11}
\tan\theta_-=\frac{\sqrt{2}|A_1|}{A_{n/2}},
\end{equation}
where $0\leq\theta_-\leq\pi$. We can uniquely recover the polar $n$-complex number given its modulus and the $n-1$ angles defined above.\footnote{Exact formulas can be found in \cite[pp.~212--216]{Olariu02}, especially (6.80), (6.81), (6.103), and (6.104). We remark that Olariu's choice of $A_1$ as a reference for the planar and polar angles is convenient but somewhat arbitrary.} More importantly, the polar $n$-complex numbers are ring-isomorphic\footnote{A ring isomorphism is a bijective map $\chi:R\rightarrow S$ such that $\chi(1_R)=1_S$, $\chi(ab)=\chi(a)\chi(b)$, and $\chi(a+b)=\chi(a)+\chi(b)$ for all $a,b\in R$.} to the following matrix representation \cite{Olariu02}, $\chi:\mathbb{K}_n\rightarrow\mathbb{R}^{n\times n}$:
\begin{equation}
\label{eq:12}
\chi(p)=\begin{bmatrix}
a_0 & a_{n-1} & a_{n-2} & \cdots & a_1\\
a_1 & a_0 & a_{n-1} & \cdots & a_2\\
a_2 & a_1 & a_0 & \cdots & a_3\\
\vdots & \vdots & \vdots & \ddots & \vdots\\
a_{n-1} & a_{n-2} & a_{n-3} & \cdots & a_0
\end{bmatrix},
\end{equation}
which is a circulant matrix.\footnote{A circulant matrix is a matrix $\mathbf{C}$ where each column is a cyclic shift of its previous column, such that $\mathbf{C}$ is diagonalizable by the DFT \cite{Davis79}. More concisely, we can write $c_{ik}=a_{(i-k)\bmod n}$.} This means that polar $n$-complex multiplication is equivalent to circular convolution. Due to the circular convolution theorem, it can be implemented efficiently in the Fourier domain \cite{Gleich13}:
\begin{equation}
\label{eq:13}
\mathbf{F}_n(\mathbf{a}\circledast\mathbf{b})=\sqrt{n}(\mathbf{F}_n\mathbf{a})\circ(\mathbf{F}_n\mathbf{b}),
\end{equation}
where $\mathbf{a},\mathbf{b}\in\mathbb{R}^n$, $\circledast$ denotes circular convolution, and $\circ$ is the Hadamard product. The isomorphism in \eqref{eq:12} implies \cite{Gleich13}:
\begin{gather}
\label{eq:14}
\chi(1)=\mathbf{I}_n,\\
\label{eq:15}
\chi(pq)=\chi(p)\chi(q),\\
\label{eq:16}
\chi(p+q)=\chi(p)+\chi(q),\\
\label{eq:17}
\chi(p^{-1})=\chi(p)^{-1},
\end{gather}
for $1,p,q\in\mathbb{K}_n$. From these properties it becomes natural to define the polar $n$-complex conjugation $\bar p$ by \cite{Gleich13}
\begin{equation}
\label{eq:18}
\chi(\bar p)=\chi(p)^*
\end{equation}
where $\chi(p)^*$ denotes the conjugate transpose of $\chi(p)$. This allows us to propose a new scalar product inspired by its quaternionic counterpart \cite{Mandic11},
\begin{equation}
\label{eq:19}
\langle p,q\rangle=\real p\bar q,
\end{equation}
which we will use later for the Frobenius norm of the polar $n$-complex numbers. Note that this differs from the usual definition $\langle p,q\rangle=p\bar q$ \cite{Gleich13} because we need the real restriction for the desirable property $\langle p,p\rangle=|p|^2$. To wit, observe that $\real p=a_0=[\chi(p)]_{ii}$ for arbitrary $i$, thus $\real p\bar q=[\chi(p)\chi(q)^*]_{ii}=\sum_{k=1}^n [\chi(p)]_{ik}[\chi(q)]_{ik}$ which is the standard inner product between the underlying elements. The same results can also be obtained from $\real\bar pq$. In other words, if $p=\sum_{i=0}^{n-1}a_ie_i$ and $q=\sum_{i=0}^{n-1}b_ie_i$, we get
\begin{equation}
\label{eq:20}
\real p\bar q=\real\bar pq=\sum_{i=0}^{n-1}a_ib_i.
\end{equation}

An alternative way of looking at the isomorphism in \eqref{eq:12} is to consider the circulant matrix as a sum \cite{Kra12},
\begin{equation}
\label{eq:21}
\chi(p)=a_0\mathbf{E}_n^0+a_1\mathbf{E}_n^1+\ldots+a_{n-1}\mathbf{E}_n^{n-1},
\end{equation}
where
\begin{equation}
\label{eq:22}
\mathbf{E}_n=\begin{bmatrix}
0 & 0 & \cdots & 0 & 1\\
1 & 0 & \cdots & 0 & 0\\
0 & 1 & \cdots & 0 & 0\\
\vdots & \vdots & \ddots & \vdots & \vdots\\
0 & 0 & \cdots & 1 & 0
\end{bmatrix}\in\mathbb{R}^{n\times n},
\end{equation}
following the convention that $\mathbf{E}_n^0=\mathbf{I}_n$. It is trivial to show that $\mathbf{E}_n^i\mathbf{E}_n^k=\mathbf{E}_n^{(i+k)\bmod n}$ \cite{Kra12}. Hence the isomorphism is immediately obvious. Recall that the group of imaginary units $\{\mathbf{E}_n^i\}_{i=0}^{n-1}$ is called cyclic if we can use a single basis element $\mathbf{E}_n$ to generate the entire algebra, so the algebra in \eqref{eq:21} has another name called a cyclic algebra \cite{Cohn03}.

The circulant isomorphism helps us to utilize recent literature on circulant algebra matrices \cite{Gleich13}, which simplifies our work in the next subsection. The circulant algebra in \cite{Gleich13} breaks the modulus into $n$ pieces such that the original number can be uniquely recovered without the planar and polar angles. However, for the $\ell_1$-norm at least, we need a single number for minimization purposes. Moreover, although our goal is phase preservation, we do not need to calculate the angles explicitly for the PCP problem. Consequently, we will stick with the original definition in \eqref{eq:5}.

\subsection{Polar $n$-Complex Matrices and Their Isomorphisms}
\label{subsec:iso}

We denote the set of $l\times m$ matrices with polar $n$-complex entries by $\mathbb{K}_n^{l\times m}$. For a polar $n$-complex matrix $\mathbf{A}\in\mathbb{K}_n^{l\times m}$, we define its adjoint matrix via $\chi_{lm}:\mathbb{K}_n^{l\times m}\rightarrow\mathbb{R}^{ln\times mn}$ \cite{Gleich13}:
\begin{equation}
\label{eq:23}
\chi_{lm}(\mathbf{A})=\begin{bmatrix}
\chi(A_{11}) & \chi(A_{12}) & \ldots & \chi(A_{1m})\\
\chi(A_{21}) & \chi(A_{22}) & \ldots & \chi(A_{2m})\\
\vdots & \vdots & \ddots & \vdots\\
\chi(A_{l1}) & \chi(A_{l2}) & \ldots & \chi(A_{lm})
\end{bmatrix}.
\end{equation}

\begin{table*}[!t]
\renewcommand{\arraystretch}{1.1}
\caption[]{Step-by-step illustration of the cft for $\mathbf{A}\in\mathbb{K}_2^{2\times 2}$; see \eqref{eq:7}, \eqref{eq:23} and \eqref{eq:24} for definitions. In general, due to the properties of the circulant blocks, the cft can block diagonalize the adjoint matrix of any polar $n$-complex matrices. Here $\mathbf{F}_2=\frac{1}{\sqrt{2}}
\begin{bmatrix}
1 & 1\\
1 & -1
\end{bmatrix}$}
\label{tab:1}
\centering
\begin{tabular}{cccc}
\hline
$\mathbf{A}$ & $\chi_{2,2}(\mathbf{A})$ & $(\mathbf{I}_2\otimes\mathbf{F}_2)\chi_{2,2}(\mathbf{A})(\mathbf{I}_2\otimes\mathbf{F}_2^*)$ & $\mathbf{P}_{4,2}(\mathbf{I}_2\otimes\mathbf{F}_2)\chi_{2,2}(\mathbf{A})(\mathbf{I}_2\otimes\mathbf{F}_2^*)\mathbf{P}_{4,2}^{-1}$\\
\hline
\noalign{\smallskip}
$\begin{bmatrix}
a_0+a_1e_1 & c_0+c_1e_1\\
b_0+b_1e_1 & d_0+d_1e_1
\end{bmatrix}$ & $\begin{bmatrix}
a_0 & a_1 & c_0 & c_1\\
a_1 & a_0 & c_1 & c_0\\
b_0 & b_1 & d_0 & d_1\\
b_1 & b_0 & d_1 & d_0
\end{bmatrix}$ & $\begin{bmatrix}
a_0+a_1 & 0 & c_0+c_1 & 0\\
0 & a_0-a_1 & 0 & c_0-c_1\\
b_0+b_1 & 0 & d_0+d_1 & 0\\
0 & b_0-b_1 & 0 & d_0-d_1
\end{bmatrix}$ & $\begin{bmatrix}
a_0+a_1 & c_0+c_1 & 0 & 0\\
b_0+b_1 & d_0+d_1 & 0 & 0\\
0 & 0 & a_0-a_1 & c_0-c_1\\
0 & 0 & b_0-b_1 & d_0-d_1
\end{bmatrix}$\\
\noalign{\smallskip}
\hline
\end{tabular}
\end{table*}

We will now show that the $\mathbb{R}$-linear map $\chi_{lm}(\mathbf{A}):\mathbb{R}^{mn}\rightarrow\mathbb{R}^{ln}$ operates in an identical manner as the $\mathbb{K}_n$-linear map $\mathbf{A}:\mathbb{K}_n^m\rightarrow\mathbb{K}_n^l$.

\begin{thm}
\label{thm:1}
Let $\mathbf{A}\in\mathbb{K}_n^{l\times m}$. Then the following holds:
\begin{enumerate}
\item $\chi_{mm}(\mathbf{I}_m)=\mathbf{I}_{mn}$ if $\mathbf{I}_m\in\mathbb{K}_n^{m\times m}$;
\item $\chi_{lr}(\mathbf{A}\mathbf{B})=\chi_{lm}(\mathbf{A})\chi_{mr}(\mathbf{B})$ if $\mathbf{B}\in\mathbb{K}_n^{m\times r}$;
\item $\chi_{lm}(\mathbf{A}+\mathbf{B})=\chi_{lm}(\mathbf{A})+\chi_{lm}(\mathbf{B})$ if $\mathbf{B}\in\mathbb{K}_n^{l\times m}$;
\item $\chi_{lm}(\mathbf{A}^*)=\chi_{lm}(\mathbf{A})^*$;
\item $\chi_{lm}(\mathbf{A}^{-1})=\chi_{lm}(\mathbf{A})^{-1}$ if it exists.
\end{enumerate}
\end{thm}

\begin{proof}
1, 3, and 4 can be verified by direct substitution. 5 can be derived from 1--2 via the equality $\mathbf{AA}^{-1}=\mathbf{I}$. 2 can be proven using \eqref{eq:12} and \eqref{eq:15}:
\begin{align*}
\chi_{lm}(\mathbf{AB}) &= \chi\left(\begin{bmatrix}
\sum\limits_{k=1}^m A_{1k}B_{k1} & \cdots & \sum\limits_{k=1}^m A_{1k}B_{kr}\\
\vdots & \ddots & \vdots\\
\sum\limits_{k=1}^m A_{lk}B_{k1} & \cdots & \sum\limits_{k=1}^m A_{lk}B_{kr}
\end{bmatrix}\right)\\
&= \begin{bmatrix}
\chi\left(\sum\limits_{k=1}^m A_{1k}B_{k1}\right) & \cdots & \chi\left(\sum\limits_{k=1}^m A_{1k}B_{kr}\right)\\
\vdots & \ddots & \vdots\\
\chi\left(\sum\limits_{k=1}^m A_{lk}B_{k1}\right) & \cdots & \chi\left(\sum\limits_{k=1}^m A_{lk}B_{kr}\right)
\end{bmatrix}\\
&= \begin{bmatrix}
\sum\limits_{k=1}^m \chi(A_{1k})\chi(B_{k1}) & \cdots & \sum\limits_{k=1}^m \chi(A_{1k})\chi(B_{kr})\\
\vdots & \ddots & \vdots\\
\sum\limits_{k=1}^m \chi(A_{lk})\chi(B_{k1}) & \cdots & \sum\limits_{k=1}^m \chi(A_{lk})\chi(B_{kr})
\end{bmatrix}\\
&= \chi_{lm}(\mathbf{A})\chi_{lm}(\mathbf{B}).
\end{align*}
In other words, the adjoint matrix $\chi_{lm}(\mathbf{A})$ is an isomorphic representation of the polar $n$-complex matrix $\mathbf{A}$.
\end{proof}

The above isomorphism is originally established for circulant matrix-vector multiplication \cite{Gleich13}, which we have just extended to the case of matrix-matrix multiplication. This isomorphism simplifies our work both theoretically and experimentally by allowing us to switch to the adjoint matrix representation where it is more convenient.

\subsection{Singular Value Decomposition}
\label{subsec:svd}

For the SVD of $\mathbf{A}\in\mathbb{K}_n^{l\times m}$, we first define the stride-by-$s$ \cite{Granata92} permutation matrix of order $m$ by:
\begin{equation}
\label{eq:24}
\left[\mathbf{P}_{m,s}\right]_{ik}=\left[\mathbf{I}_m\right]_{is-(m-1)\lfloor is/m\rfloor,k}
\end{equation}
for $i,k=0,1,\ldots,m-1$. This is equivalent to but more succinct than the standard definition in the literature \cite{Granata92}. The stride-by-$s$ permutation greatly simplifies the definition of the two-dimensional shuffle in the following. We define the circulant Fourier transform (CFT) and its inverse (ICFT), in the same way as \cite{Gleich13}:
\begin{gather}
\label{eq:25}
\mathrm{cft}(\mathbf{A})=\mathbf{P}_{ln,l}(\mathbf{I}_l\otimes\mathbf{F}_n)\chi_{lm}(\mathbf{A})(\mathbf{I}_m\otimes\mathbf{F}_n^*)\mathbf{P}_{mn,m}^{-1},\\
\label{eq:26}
\chi_{lm}(\mathrm{icft}(\mathbf{\hat{A}}))=(\mathbf{I}_l\otimes\mathbf{F}_n^*)\mathbf{P}_{ln,l}^{-1}\mathbf{\hat{A}}\mathbf{P}_{mn,m}(\mathbf{I}_m\otimes\mathbf{F}_n),
\end{gather}
where $\mathbf{P}_{ln,l}(\cdot)\mathbf{P}_{mn,m}^{-1}$ shuffles an $ln\times mn$ matrix containing $n\times n$ diagonal blocks into a block diagonal matrix containing $l\times m$ blocks. Please refer to Table~\ref{tab:1} to see this shuffle in action. The purpose of $\mathrm{cft}(\mathbf{A})$ is to block diagonalize the adjoint matrix of $\mathbf{A}$ into the following form \cite{Gleich13}:
\begin{equation}
\label{eq:27}
\mathbf{\hat{A}}=\mathrm{cft}(\mathbf{A})=\begin{bmatrix}
\mathbf{\hat{A}}_1\\
& \ddots\\
& & \mathbf{\hat{A}}_n
\end{bmatrix},
\end{equation}
while $\mathrm{icft}(\mathbf{\hat{A}})$ inverts this operation. Here, $\mathbf{\hat{A}}_i$ can be understood as the eigenvalues of the input as produced in the Fourier transform order, as noted by \cite{Gleich13}. The SVD of $\mathbf{A}$ can be performed blockwise through the SVD of $\mathrm{cft}(\mathbf{A})$ \cite{Kilmer08}:
\begin{equation}
\label{eq:28}
\begin{bmatrix}
\mathbf{\hat{U}}_1\\
& \ddots\\
& & \mathbf{\hat{U}}_n
\end{bmatrix}
\begin{bmatrix}
\mathbf{\hat{\Sigma}}_1\\
& \ddots\\
& & \mathbf{\hat{\Sigma}}_n
\end{bmatrix}
\begin{bmatrix}
\mathbf{\hat{V}}_1\\
& \ddots\\
& & \mathbf{\hat{V}}_n
\end{bmatrix}^*,
\end{equation}
then we can use $\mathrm{icft}(\mathbf{\hat{U}})$, $\mathrm{icft}(\mathbf{\hat{\Sigma}})$, and $\mathrm{icft}(\mathbf{\hat{V}})$ to get $\mathbf{U}\in\mathbb{K}_n^{l\times l}$, $\mathbf{S}\in\mathbb{K}_n^{l\times m}$, and $\mathbf{V}\in\mathbb{K}_n^{m\times m}$ where $\mathbf{U}$ and $\mathbf{V}$ are unitary \cite{Kilmer08,Gleich13}. This is equivalent to the t-SVD in tensor signal processing (see Algorithm \ref{alg:1}) \cite{Kilmer08}, provided that we store the $l\times m$ polar $n$-complex matrix into an $l\times m\times n$ real tensor,\footnote{By convention, we denote tensors with calligraphic letters. For a three-dimensional tensor $\mathcal{A}\in\mathbb{R}^{n_1\times n_2\times n_3}$, a fiber is a one-dimension subarray defined by fixing two of the indices, whereas a slice is a two-dimensional subarray defined by fixing one of the indices \cite{Kolda09}. The $(i,k,l)$-th element of $\mathcal{A}$ is denoted by $A_{ikl}$. If we indicate all elements of a one-dimensional subarray using the \textsc{Matlab} colon notation, then $\mathbf{A}_{:kl}$, $\mathbf{A}_{i:l}$, and $\mathbf{A}_{ik:}$ are called the column, row and tube fibers, respectively \cite{Kolda09}. Similarly, $\mathbf{A}_{i::}$, $\mathbf{A}_{:k:}$, and $\mathbf{A}_{::l}$ are called the horizontal, lateral, and frontal slides, respectively \cite{Kolda09}. Notably, Kilmer, Martin, and Perrone \cite{Kilmer08} reinterprets an $n_1\times n_2\times n_3$ tensor as an $n_1\times n_2$ matrix of tubes (of length $n_3$). This is most relevant to our present work when polar $n_3$-complex numbers are seen as tubes.} then the $n$-point DFT along all tubes is equivalent to the CFT. Matrix multiplication can also be done blockwise in the CFT domain with the $\sqrt{n}$ scaling as before.

\begin{algorithm}
\caption{t-SVD \cite{Kilmer08}}
\begin{algorithmic}[1]
\label{alg:1}
\REQUIRE $\mathcal{X}\in\mathbb{C}^{l\times m\times n}$ \COMMENT{See footnote 5 for tensor notation.}
\ENSURE $\mathcal{U}$, $\mathcal{S}$, $\mathcal{V}$
\STATE ${\mathcal{\hat{X}}}\leftarrow\rm{fft}(\mathcal{X},n,3)$ \COMMENT{Applies $n$-point DFT to each tube.}
\FOR{$i=1:n$}
    \STATE $[\mathbf{\hat{U}}_{::i},\mathbf{\hat{S}}_{::i},\mathbf{\hat{V}}_{::i}]\leftarrow\mathrm{svd}(\mathbf{\hat{X}}_{::i})$ \COMMENT{SVD each frontal slide.}
\ENDFOR
\STATE $\mathcal{U}\leftarrow\mathrm{ifft}(\mathcal{\hat{U}},n,3);\ \mathcal{S}\leftarrow\mathrm{ifft}(\mathcal{\hat{S}},n,3);\ \mathcal{V}\leftarrow\mathrm{ifft}(\mathcal{\hat{V}},n,3)$
\end{algorithmic}
\end{algorithm}

\subsection{Proposed Extensions}
\label{subsec:iso}

In order to study the phase angle between matrices, we define a new polar $n$-complex inner product as
\begin{equation}
\label{eq:29}
\langle\mathbf{A},\mathbf{B}\rangle=\real\mathrm{tr}(\mathbf{AB}^*),\ \mathbf{A},\mathbf{B}\in\mathbb{K}_n^{l\times m}.
\end{equation}
and use it to induce the polar $n$-complex Frobenius norm:
\begin{equation}
\label{eq:30}
\|\mathbf{A}\|_F=\sqrt{\langle\mathbf{A},\mathbf{A}\rangle}.
\end{equation}
We propose two further isomorphisms for polar $n$-complex matrices via $\xi:\mathbb{K}_n^{l\times m}\rightarrow\mathbb{R}^{l\times mn}$ and $\nu:\mathbb{K}_n^{l\times m}\rightarrow\mathbb{R}^{lmn}$:
\begin{gather}
\label{eq:31}
\xi(\mathbf{A})=[\imag_0\mathbf{A},\imag_1\mathbf{A},\ldots,\imag_{n-1}\mathbf{A}],\\
\nu(\mathbf{A})=\vect\xi(\mathbf{A}).
\end{gather}
These are the polar $n$-complex matrix counterparts of the tensor $\mathrm{unfold}$ and $\vect$ operators, respectively.\footnote{Column unfolding reshapes the tensor $\mathcal{A}\in\mathbb{R}^{n_1\times n_2\times n_3}$ into a matrix $\mathbf{M}\in\mathbb{R}^{n_1\times n_2n_3}$ by mapping each tensor element $A_{ikl}$ into the corresponding matrix element $M_{i,k+(l-1)n_2}$ \cite{Kolda09}.} We end this subsection by enumerating two elementary algebraic properties of $\mathbb{K}_n^{l\times m}$, which will come in handy when we investigate the trace norm later in Theorem \ref{thm:9}. The proofs are given below for completeness.

\begin{prop}
\label{thm:2}
If $\mathbf{A},\mathbf{B}\in\mathbb{K}_n^{l\times m}$, then the following holds:
\begin{enumerate}
\item $\langle\mathbf{A},\mathbf{B}\rangle=\real\mathrm{tr}(\mathbf{A}^*\mathbf{B})=\nu(\mathbf{A})^T\nu(\mathbf{B})$;
\item $\|\mathbf{A}\|_F^2=\sum_i|\sigma_i(\mathbf{A})|$.
\end{enumerate}
where $\sigma_i(\mathbf{A})$ are the singular values of $\mathbf{A}$ obtained from $\mathrm{icft}(\mathbf{\hat{\Sigma}})$ after steps \eqref{eq:27} and \eqref{eq:28}.
\end{prop}

\begin{proof}
\leavevmode
\begin{enumerate}
\item This is a direct consequence of (20) after observing that $\real\mathrm{tr}(\mathbf{AB}^*)=\real\sum_{i,k}A_{ik}\bar B_{ik}$. From this we can say that our polar $n$-complex inner product is Euclidean. As a corollary we have $\|\mathbf{A}\|_F^2=\sum_{i,k}|A_{ik}|^2$.
\item As the Frobenius norm is invariant under any unitary transformation \cite{Horn13}, we can write $\|\mathbf{A}\|_F^2=\|\mathbf{\Sigma}\|_F^2=\sum_i|\sigma_i(\mathbf{A})|^2$.
\end{enumerate}
\end{proof}

\section{Extension to Polar $n$-Bicomplex Numbers}
\label{sec:n-bicomplex}

One problem with the real numbers is that $\sqrt{-1}\notin\mathbb{R}$; that is, they are not algebraically closed. This affects the polar $n$-complex numbers too since their real and imaginary parts consist of real coefficients only. To impose algebraic closure for certain applications, we can go one step further and use complex coefficients instead. More specifically, we extend the polar $n$-complex algebra by allowing for complex coefficients in \eqref{eq:3}, such that
\begin{equation}
\label{eq:33}
p=a_0e_0+a_1e_1+\cdots+a_{n-1}e_{n-1}\in\mathbb{CK}_n,
\end{equation}
where $a_0,a_1,\ldots,a_{n-1}\in\mathbb{C}$. In other words, both real and imaginary parts of $p$ now contain complex numbers (effectively doubling its dimensions). This constitutes our definition of the polar $n$-bicomplex numbers $\mathbb{CK}_n$. The first imaginary unit is still $e_0=1$ and $e_1,\ldots,e_{n-1}$ satisfies the same multiplication table in \eqref{eq:4}. We can now write $\real p=\real a_0$ for the real part of $p$ (note the additional $\real$) and $\imag_i p=a_i$ for the imaginary parts for $i=0,1,\ldots,n-1$ (as before, the imaginary part includes the real part for notational convenience). The modulus then becomes
\begin{equation}
\label{eq:34}
|p|=\sqrt{|a_0|^2+|a_1|^2+\cdots+|a_{n-1}|^2},
\end{equation}
along with the same $n-1$ angles in (\ref{eq:8}--\ref{eq:11}). For example, if $g=(1+2j)+(3+4j)e_1+(5+6j)e_2$, we have $\real g=1$, $\imag_0 g=1+2j$, $\imag_1 g=3+4j$, $\imag_2 g=5+6j$, and $|g|=\sqrt{91}$. The polar $n$-bicomplex numbers are ring-isomorphic to the same matrix in \eqref{eq:12}, and have the same properties (\ref{eq:14}--\ref{eq:17}). The multiplication can still be done in the Fourier domain if desired. The polar $n$-bicomplex conjugation can be defined in the same manner as \eqref{eq:18}. Given our new definition of $\real$, the scalar product is:
\begin{equation}
\label{eq:35}
\langle p,q\rangle=\real p\bar q.
\end{equation}
Note that we still have $\langle p,p\rangle=|p|^2$, because $\real p\bar q=\real [\chi(p)\chi(q)^*]_{ii}=\real\sum_{k=1}^n [\chi(p)]_{ik}[\overline{\chi(q)}]_{ik}$ for arbitrary $i$, which gives the Euclidean inner product (likewise for $\real\bar pq$). So given $p=\sum_{i=0}^{n-1}a_ie_i$ and $q=\sum_{i=0}^{n-1}b_ie_i$, we now have
\begin{equation}
\label{eq:36}
\real p\bar q=\real\bar pq=\real\sum_{i=0}^{n-1}a_i\bar b_i.
\end{equation}

\subsection{Polar $n$-Bicomplex Matrices and Their Isomorphisms}
\label{subsec:bi-iso}

Analogously, we denote the set of $l\times m$ matrices with polar $n$-bicomplex entries by $\mathbb{CK}_n^{l\times m}$. The adjoint matrix of $\mathbf{A}\in\mathbb{CK}_n^{l\times m}$ can be defined similarly via $\chi_{lm}:\mathbb{CK}_n^{l\times m}\rightarrow\mathbb{C}^{ln\times mn}$:
\begin{equation}
\label{eq:37}
\chi_{lm}(\mathbf{A})=\begin{bmatrix}
\chi(A_{11}) & \chi(A_{12}) & \ldots & \chi(A_{1m})\\
\chi(A_{21}) & \chi(A_{22}) & \ldots & \chi(A_{2m})\\
\vdots & \vdots & \ddots & \vdots\\
\chi(A_{l1}) & \chi(A_{l2}) & \ldots & \chi(A_{lm})
\end{bmatrix}.
\end{equation}

Next we are going to show that the $\mathbb{C}$-linear map $\chi_{lm}(\mathbf{A}):\mathbb{C}^{mn}\rightarrow\mathbb{C}^{ln}$ operates in the same manner as the $\mathbb{CK}_n$-linear map $\mathbf{A}:\mathbb{CK}_n^m\rightarrow\mathbb{CK}_n^l$.

\begin{thm}
\label{thm:3}
Let $\mathbf{A}\in\mathbb{CK}_n^{l\times m}$. Then we have:
\begin{enumerate}
\item $\chi_{mm}(\mathbf{I}_m)=\mathbf{I}_{mn}$ if $\mathbf{I}_m\in\mathbb{CK}_n^{m\times m}$;
\item $\chi_{lr}(\mathbf{A}\mathbf{B})=\chi_{lm}(\mathbf{A})\chi_{mr}(\mathbf{B})$ if $\mathbf{B}\in\mathbb{CK}_n^{m\times r}$;
\item $\chi_{lm}(\mathbf{A}+\mathbf{B})=\chi_{lm}(\mathbf{A})+\chi_{lm}(\mathbf{B})$ if $\mathbf{B}\in\mathbb{CK}_n^{l\times m}$;
\item $\chi_{lm}(\mathbf{A}^*)=\chi_{lm}(\mathbf{A})^*$;
\item $\chi_{lm}(\mathbf{A}^{-1})=\chi_{lm}(\mathbf{A})^{-1}$ if it exists.
\end{enumerate}
\end{thm}

\begin{proof}
See Theorem \ref{thm:1}.
\end{proof}

The polar $n$-bicomplex SVD, inner product and Frobenius norm can be defined following \eqref{eq:28}, \eqref{eq:29} and \eqref{eq:30}. The illustration in Table~\ref{tab:1} still applies. The additional isomorphisms are defined via $\xi:\mathbb{CK}_n^{l\times m}\rightarrow\mathbb{R}^{l\times 2mn}$ and $\nu:\mathbb{CK}_n^{l\times m}\rightarrow\mathbb{R}^{2lmn}$:
\begin{gather}
\label{eq:38}
\xi(\mathbf{A})=[\real\imag_0\mathbf{A},\imag\imag_0\mathbf{A},\ldots,\imag\imag_{n-1}\mathbf{A}]\\
\label{eq:39}
\nu(\mathbf{A})=\vect\xi(\mathbf{A}).
\end{gather}

\begin{prop}
\label{thm:4}
If $\mathbf{A},\mathbf{B}\in\mathbb{CK}_n^{l\times m}$, then the following holds:
\begin{enumerate}
\item $\langle\mathbf{A},\mathbf{B}\rangle=\real\mathrm{tr}(\mathbf{A}^*\mathbf{B})=\nu(\mathbf{A})^T\nu(\mathbf{B})$;
\item $\|\mathbf{A}\|_F^2=\sum_i|\sigma_i(\mathbf{A})|$,
\end{enumerate}
where $\sigma_i(\mathbf{A})$ are the singular values of $\mathbf{A}$.
\end{prop}

\begin{proof}
See Proposition \ref{thm:2}.
\end{proof}

\section{Polar $n$-Complex and $n$-Bicomplex PCP}
\label{sec:pcp}

PCP algorithms \cite{Candes11,Huang12} are traditionally implemented by proximal optimization \cite{Combettes11} which extends gradient projection to the nonsmooth case. Often, closed-form solutions for the proximity operators are available, like soft-thresholding \cite{Donoho95} and singular value thresholding \cite{Cai10} in the real-valued case.

\subsection{Equivalence to Real-Valued Proximal Methods}
\label{subsec:equiv}

To fix our notation, recall that the proximity operator of a function $f:\mathbb{R}^m\rightarrow\mathbb{R}^m$ is traditionally defined as \cite{Combettes11}:
\begin{equation}
\label{eq:40}
\mathrm{prox}_f\mathbf{z}=\arg\min_\mathbf{x}\frac{1}{2}\|\mathbf{z}-\mathbf{x}\|_2^2+f(\mathbf{x}),\ \mathbf{x}\in\mathbb{R}^m.
\end{equation}
For $\mathbf{x}\in\mathbb{K}_n^m$ or $\mathbb{CK}_n^m$ we can use $\nu(\mathbf{x})$ instead of $\mathbf{x}$ and adjust $f(\mathbf{x})$ accordingly. As $\|\mathbf{z}-\mathbf{x}\|_2^2$ is invariant under this transformation, we can equivalently extend the domain of $f$ to $\mathbb{K}_n^m$ or $\mathbb{CK}_n^m$ without adjusting $f(\mathbf{x})$ in the following. This equivalence establishes the validity of proximal minimization using polar $n$-complex and $n$-bicomplex matrices directly, without needing to convert to the real domain temporarily.

\subsection{The Proximity Operator for the $\ell_1$ Norm}
\label{subsec:prox_1}

We deal with the $\ell_1$- and trace-norm regularizers in order.
\begin{lem}[Yuan and Lin \cite{Yuan06}]
\label{lem:5}
Let \{$\mathbf{x}_{(1)},\ldots,\mathbf{x}_{(m)}$\} be a partition of $\mathbf{x}$ such that $\mathbf{x}=\bigcup_{i=1}^m\,\mathbf{x}_{(i)}$. The proximity operator for the group lasso regularizer $\lambda\sum_{i=1}^m\|\mathbf{x}_{(i)}\|_2$ is
\begin{equation}
\label{eq:41}
\mathrm{prox}_{\lambda\sum\|\cdot\|_2}\mathbf{z}=\left[\left(1-\frac{\lambda}{\|\mathbf{z}_{(i)}\|_2}\right)_+\mathbf{z}_{(i)}\right]_{i=1}^m,\ \mathbf{z}\in\mathbb{R}^r,
\end{equation}
where $(\mathbf{y})_+$ denotes $\max(0,\mathbf{y})$, $[\mathbf{y}_i]_{i=1}^m=[\mathbf{y}_1^T,\ldots,\mathbf{y}_m^T]^T$ is a real column vector, and $r$ is the sum of the sizes of $\mathbf{x}_{(\cdot)}$.
\end{lem}

\begin{proof}
This result is standard in sparse coding \cite{Yuan06,Tomioka12}.
\end{proof}

The group lasso is a variant of sparse coding that promotes group sparsity, i.e., zeroing entire groups of variables at once or not at all. When we put the real and imaginary parts of a polar $n$-complex or $n$-bicomplex number in the same group, group sparsity makes sense, since a number cannot be zero unless all its constituent parts are zero, as in the next theorem.

\begin{thm}
\label{thm:6}
The polar $n$-complex or $n$-bicomplex lasso
\begin{equation}
\label{eq:42}
\min_\mathbf{x}\frac{1}{2}\|\mathbf{z}-\mathbf{x}\|_2^2+\lambda\|\mathbf{x}\|_1,\ \mathbf{z},\mathbf{x}\in F^m,
\end{equation}
where $F$ is $\mathbb{K}_n$ or $\mathbb{CK}_n$, is equivalent to the group lasso
\begin{equation}
\label{eq:43}
\min_{\xi(\mathbf{x})}\frac{1}{2}\|\xi(\mathbf{z-x})\|_F^2+\lambda\|\xi(\mathbf{x})\|_{1,2},
\end{equation}
where $\|\mathbf{A}\|_{1,2}$ is defined as $\sum_i\sqrt{\sum_k|A_{ik}|^2}$.
\end{thm}

\begin{proof}
The proof is straightforward:
\begin{align*}
\frac{1}{2}\|\mathbf{z}-\mathbf{x}\|_2^2+\lambda\|\mathbf{x}\|_1&=\sum_i\frac{1}{2}|z_i-x_i|^2+\lambda|x_i|\\
&=\sum_i\frac{1}{2}\|\xi(z_i-x_i)\|_2^2+\lambda\|\xi(x_i)\|_2\\
&=\frac{1}{2}\|\xi(\mathbf{z}-\mathbf{x})\|_F^2+\lambda\|\xi(\mathbf{x})\|_{1,2}.
\end{align*}
\end{proof}

The first line invokes the definitions of $|\cdot|$ in \eqref{eq:5} and \eqref{eq:34}, while the second line is due to the proposed isomorphisms in \eqref{eq:31} and \eqref{eq:38}. In other words, we have discovered a method to solve the novel polar $n$-complex or $n$-bicomplex lasso problem using real-valued group lasso solvers. By combining Lemma~\ref{lem:5} and Theorem~\ref{thm:6}, we arrive at the main result of this subsection.

\begin{cor}
\label{cor:7}
For the entrywise $\ell_1$-regularizer $\lambda\|X\|_1$, where $X,Z\in\mathbb{K}_n^{l\times m}$ or $\mathbb{CK}_n^{l\times m}$, we may treat $X$ as a long hypercomplex vector of length $lm$ without loss of generality. Simply assign each hypercomplex number to its own group $\mathbf{g}_i$, for all $1\leq i\leq lm$ numbers, and we obtain the proximity operator for $\lambda\|X\|_1$ using \eqref{eq:41}:
\begin{equation}
\label{eq:44}
\mathrm{prox}_{\lambda\|\cdot\|_1}^F\mathbf{z}=\left(1-\frac{\lambda}{|\mathbf{z}|}\right)_+\mathbf{z},\ \mathbf{z}\in F^{lm},
\end{equation}
where $F$ is $\mathbb{K}_n$ or $\mathbb{CK}_n$ and $\mathbf{z}=\vect Z$. Here $|\mathbf{z}|$ corresponds to the Euclidean norm in \eqref{eq:41} and the grouping should follow the definition of $\xi(\mathbf{A})$ for the respective algebra. Note how each entry corresponds to its real-isomorphic group $\xi(\cdot)$ here.
\end{cor}

\subsection{The Proximity Operator for the Trace Norm}
\label{subsec:prox_*}

Next we will treat the trace-norm regularizer. We begin our proof by quoting a classic textbook inequality. In what follows, $\sigma_i(\mathbf{A})$ denotes the singular values of $\mathbf{A}$.

\begin{lem}[von Neumann \cite{Horn13}]
\label{lem:8}
For any $\mathbf{A},\mathbf{B}\in\mathbb{C}^{l\times m}$, the von Neumann trace inequality holds:
\begin{equation}
\label{eq:45}
\real\mathrm{tr}(\mathbf{AB}^*)\leq\sum_i \sigma_i(\mathbf{A})\sigma_i(\mathbf{B}).
\end{equation}
\end{lem}

\begin{proof}
This is a standard textbook result \cite{Horn13}.
\end{proof}

\begin{thm}
\label{thm:9}
For any $\mathbf{A},\mathbf{B}\in\mathbb{K}_n^{l\times m}$ or $\mathbb{CK}_n^{l\times m}$, the following extension to the von Neumann inequality holds:
\begin{equation}
\label{eq:46}
\real\mathrm{tr}(\mathbf{AB}^*)\leq\sum_i|\sigma_i(\mathbf{A})||\sigma_i(\mathbf{B})|.
\end{equation}
\end{thm}

\begin{proof}
This theorem embodies the key insight of this paper. Our novel discovery is that we can switch to the block-diagonalized CFT space to separate the sums and switch back:
\begin{align*}
\real\mathrm{tr}(\mathbf{\hat{A}\hat{B}}^*)&=\sum_{k=0}^{n-1}\real\mathrm{tr}(\mathbf{\hat{A}}_k\mathbf{\hat{B}}_k^*)\\
&\leq\sum_i\sum_{k=0}^{n-1}\sigma_i(\mathbf{\hat{A}}_k)\sigma_i(\mathbf{\hat{B}}_k)\\
&\leq\sum_i\sqrt{\sum_{k=0}^{n-1}\sigma_i^2(\mathbf{\hat{A}}_k)\sum_{k=0}^{n-1}\sigma_i^2(\mathbf{\hat{B}}_k})\\
&=\sum_i|\sigma_i(\mathbf{\hat{A}})||\sigma_i(\mathbf{\hat{B}})|,
\end{align*}
where $\mathbf{\hat{A}}=\mathrm{cft}(\mathbf{A})$ and $\mathbf{\hat{B}}=\mathrm{cft}(\mathbf{B})$, respectively. The second line is by Lemma \ref{lem:8} and the third line is due to the Cauchy-Schwarz inequality. Using Parseval's theorem, this theorem is proved.
\end{proof}

\begin{thm}
\label{thm:10}
The proximity operator for the polar $n$-complex or $n$-bicomplex trace norm $\lambda\sum_i|\sigma_i(\mathbf{X})|$, assuming $X,Z\in\mathbb{K}_n^{l\times m}$ or $\mathbb{CK}_n^{l\times m}$, is:
\begin{equation}
\label{eq:47}
\mathrm{prox}_{\lambda\|\cdot\|_*}\mathbf{z}=\vect\mathbf{U}\left[\left(1-\frac{\lambda}{\mathbf{|\Sigma}|}\right)_+\circ\mathbf{\Sigma}\right]\mathbf{V}^*,\ \mathbf{z}\in F^{lm},
\end{equation}
where $\mathbf{z}=\vect\mathbf{Z}$, $\mathbf{U\Sigma V}^*$ is the SVD of $\mathbf{Z}$ with singular values $\mathbf{\Sigma}_{ii}=\sigma_i(\mathbf{Z})$, the absolute value of $\mathbf{\Sigma}$ is computed entrywise, and $F$ is $\mathbb{K}_n$ or $\mathbb{CK}_n$.
\end{thm}

\begin{proof}
The proof follows \cite{Tomioka12} closely except that Theorem \ref{thm:9} allows us to extend the proof to the polar $n$-complex and $n$-bicomplex cases. Starting from the Euclidean inner product identity $\langle z-x,z-x\rangle=\langle z,z\rangle -2\langle z,x\rangle +\langle x,x\rangle$, which is applicable because of Propositions \ref{thm:2} and \ref{thm:4}, we have the following inequality:

\begin{align*}
\left\Vert\mathbf{Z}-\mathbf{X}\right\Vert_F^2&=\sum_i|\sigma_i(\mathbf{Z})|^2-2\left\langle\mathbf{Z},\mathbf{X}\right\rangle +\sum_i|\sigma_i(\mathbf{X})|^2\\
&\geq\sum_i|\sigma_i(\mathbf{Z})|^2-2|\sigma_i(\mathbf{Z})||\sigma_i(\mathbf{X})|+|\sigma_i(\mathbf{X})|^2\\
&=\sum_i\left(|\sigma_i(\mathbf{Z})|-|\sigma_i(\mathbf{X})|\right)^2,
\end{align*}
where Theorem \ref{thm:9} is invoked on the penultimate line. Thus:
\begin{align*}
\frac{1}{2}\left\|\mathbf{Z}-\mathbf{X}\right\|_F^2&+\lambda\sum_i|\sigma_i(\mathbf{X})|\\
\geq&\sum_i\frac{1}{2}\left(|\sigma_i(\mathbf{Z})|-|\sigma_i(\mathbf{X})|\right)^2+\lambda|\sigma_i(\mathbf{X})|\\
=&\frac{1}{2}\||\boldsymbol{\sigma}(\mathbf{Z})|-|\boldsymbol{\sigma}(\mathbf{X})|\|_2^2+\lambda\|\boldsymbol{\sigma}(\mathbf{X})\|_1,
\end{align*}
which is equivalent to a lasso problem on the (elementwise) modulus of the singular values of a polar $n$-complex or $n$-bicomplex matrix. By applying Corollary~\ref{cor:7} to the modulus of the singular values entrywise, the theorem is proved.
\end{proof}

Unlike the entrywise $\ell_1$-regularizer, the proximity operator in Theorem~\ref{thm:10} first operates on the entire matrix all at once. Once the SVD is computed, the absolute value of its singular values are then calculated entrywise (or real-isomorphic groupwise) to respect the properties of the underlying algebra.

\subsection{The Extended Formulations of PCP}
\label{subsec:form}

With the new proximal operators in \eqref{eq:44} and \eqref{eq:47}, we can finally define the polar $n$-complex and $n$-bicomplex PCP:
\begin{equation}
\label{eq:48}
\min_{\mathbf{L},\mathbf{S}}\|\mathbf{L}\|_*+\lambda\|\mathbf{S}\|_1\mbox{\quad s.t.\quad}\mathbf{X}=\mathbf{L}+\mathbf{S}\,,
\end{equation}
where $\mathbf{X}\in\mathbb{K}_n^{l\times m}$ for the polar $n$-complex PCP and $\mathbf{X}\in\mathbb{CK}_n^{l\times m}$ for the polar $n$-bicomplex PCP. We can solve this by the same algorithms in \cite{Lin09}, except that we should replace the soft-thresholding function:
\begin{equation}
\label{eq:49}
\mathcal{S}_\lambda[x]=\left\{\begin{array}{ll}
x-\lambda,&\mbox{if }x>\lambda,\\
x+\lambda,&\mbox{if }x<-\lambda,\\
0,&\text{otherwise}
\end{array}\right.
\end{equation}
with $\mathrm{prox}_{\lambda\|\cdot\|_1}^{\mathbb{K}_n}\mathbf{z}$ and $\mathrm{prox}_{\lambda\|\cdot\|_1}^{\mathbb{CK}_n}\mathbf{z}$ for the polar $n$-complex and $n$-bicomplex PCP, respectively. The inexact augmented Lagrange multiplier (IALM) method, also known as alternating direction method of multipliers, is well-established in the literature and its convergence has long been proven \cite{Lions79,Eckstein92,Kontogiorgis98}. Our adaptation is shown in Algorithm \ref{alg:2}. As the constraint $\mathbf{X}=\mathbf{L}+\mathbf{S}$ only uses simple additions, which are elementwise by definition, IALM will continue to work without change (via Proposition~\ref{thm:2} and Proposition~\ref{thm:4}). In the original IALM formulation \cite{Lin09}, their choice of $\mathbf{Y}_1$ is informed by the dual problem, whereas their $\mu_k$'s are incremented geometrically to infinity. We will simply follow them here. In theory, any initial value would work, but good guesses would converge faster \cite{Lin09}. As for $\mu_k$, any increasing sequence can be used, so long as it satisfies the convergence assumptions $\sum_{k=1}^\infty \mu_{k+1}/\mu_k^2<\infty$ and $\lim_{k\rightarrow\infty} \mu_k(S_{k+1}-S_k)=0$ \cite{Lin09}. As both $\mathbb{K}_n$ and $\mathbb{CK}_n$ are isomorphic to the circulant algebra, the easiest option is to use Gleich's circulant algebra matrix (CAMAT) toolbox \cite{Gleich13} to implement the algorithms. However, CAMAT is slightly slow due to unnecessary conversions to and from frequency domain at each iteration, so we reimplement this algebra from scratch and entirely in the Fourier domain, via \eqref{eq:13}. See Algorithm \ref{alg:2b} for our optimized frequency domain implementation. The extra $\sqrt{n}$ scaling for the proximity operators is due to the fact that \textsc{Matlab}'s $\mathrm{fft}$ is unnormalized.\footnote{All the code for this paper (including Algorithms \ref{alg:2b} and \ref{alg:2c}) is available at http://mac.citi.sinica.edu.tw/ikala/code.html to support reproducibility.}

\begin{algorithm}[t]
\caption{Polar $n$-(Bi)complex PCP}
\begin{algorithmic}[1]
\label{alg:2}
\REQUIRE $\mathbf{X}\in F^{l\times m},F\in\{\mathbb{K}_n,\mathbb{CK}_n\}$, $\lambda\in\mathbb{R}$, $\boldsymbol{\mu}\in\mathbb{R}^\infty$
\ENSURE $\mathbf{L}_k$, $\mathbf{S}_k$
\STATE Let $\mathbf{S}_1=0$, $\mathbf{Y}_1=\mathbf{X}/\max\left(\|\mathbf{X}\|_2,\lambda^{-1}\|\mathbf{X}\|_\infty\right)$, $k=1$
\WHILE{not converged}
    \STATE $\mathbf{L}_{k+1}\gets\mathrm{prox}_{1/\mu_k\|\cdot\|_*}(\mathbf{X}-\mathbf{S}_k+\mu_k^{-1}\mathbf{Y}_k)$
    \STATE $\mathbf{S}_{k+1}\gets\mathrm{prox}_{\lambda/\mu_k\|\cdot\|_1}^F(\mathbf{X}-\mathbf{L}_{k+1}+\mu_k^{-1}\mathbf{Y}_k)$
    \STATE $\mathbf{Y}_{k+1}\gets\mathbf{Y}_k+\mu_k(\mathbf{X}-\mathbf{L}_{k+1}-\mathbf{S}_{k+1})$
    \STATE $k\gets k+1$
\ENDWHILE
\end{algorithmic}
\end{algorithm}

\begin{algorithm}[t]
\caption{Optimized Polar $n$-(Bi)complex PCP}
\begin{algorithmic}[1]
\label{alg:2b}
\REQUIRE $\mathbf{X}\in F^{l\times m},F\in\{\mathbb{K}_n,\mathbb{CK}_n\}$, $\lambda\in\mathbb{R}$, $\boldsymbol{\mu}\in\mathbb{R}^\infty$
\ENSURE $\mathbf{L}$, $\mathbf{S}$
\STATE Let $\mathbf{\hat{S}}=0$, $\mathbf{Y}=\mathbf{X}/\max\left(\|\mathbf{X}\|_2,\lambda^{-1}\|\mathbf{X}\|_\infty\right)$, $k=1$
\STATE ${\mathbf{\hat{X}}}\leftarrow\rm{fft}(\mathbf{X},n,3)$ \COMMENT{Applies $n$-point DFT to each tube.}
\STATE ${\mathbf{\hat{Y}}}\leftarrow\rm{fft}(\mathbf{Y},n,3)$
\WHILE{not converged}
    \STATE $\mathbf{\hat{Z}}\gets\mathbf{\hat{X}}-\mathbf{\hat{S}}+\mu_k^{-1}\mathbf{\hat{Y}}$
    \FOR{$i=1:n$}
        \STATE $[\mathbf{\hat{U}}_{::i},\mathbf{\hat{\Sigma}}_{::i},\mathbf{\hat{V}}_{::i}]\leftarrow\mathrm{svd}(\mathbf{\hat{Z}}_{::i})$
    \ENDFOR
    \STATE $\mathbf{\hat{\Sigma}}\gets\mathrm{prox}_{\sqrt{n}/\mu_k\|\cdot\|_1}^F \mathbf{\hat{\Sigma}}$
    \FOR{$i=1:n$}
        \STATE $\mathbf{\hat{L}}_{::i}=\mathbf{\hat{U}}_{::i}\mathbf{\hat{\Sigma}}_{::i}\mathbf{\hat{V}}_{::i}^*$
    \ENDFOR
    \STATE $\mathbf{\hat{S}}\gets\mathrm{prox}_{\lambda\sqrt{n}/\mu_k\|\cdot\|_1}^F(\mathbf{\hat{X}}-\mathbf{\hat{L}}+\mu_k^{-1}\mathbf{\hat{Y}})$
    \STATE $\mathbf{\hat{Y}}\gets\mathbf{\hat{Y}}+\mu_k(\mathbf{\hat{X}}-\mathbf{\hat{L}}-\mathbf{\hat{S}})$
    \STATE $k\gets k+1$
\ENDWHILE
\STATE $\mathbf{L}\leftarrow\mathrm{ifft}(\mathbf{\hat{L}},n,3)$
\STATE $\mathbf{S}\leftarrow\mathrm{ifft}(\mathbf{\hat{S}},n,3)$
\end{algorithmic}
\end{algorithm}

\begin{figure*}
\centering
\subfloat[]{\includegraphics[width=.55\columnwidth]{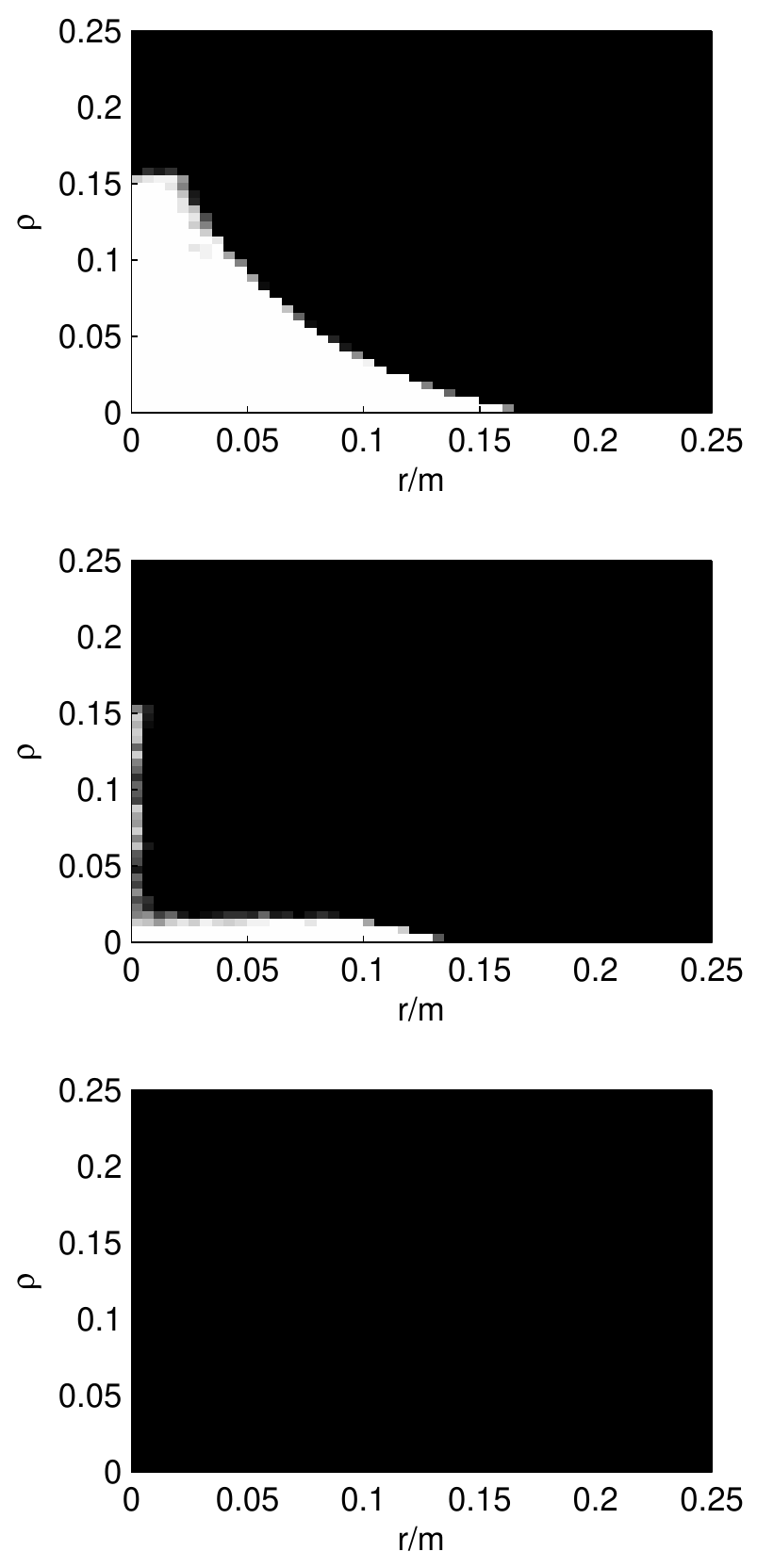}}
\subfloat[]{\includegraphics[width=.55\columnwidth]{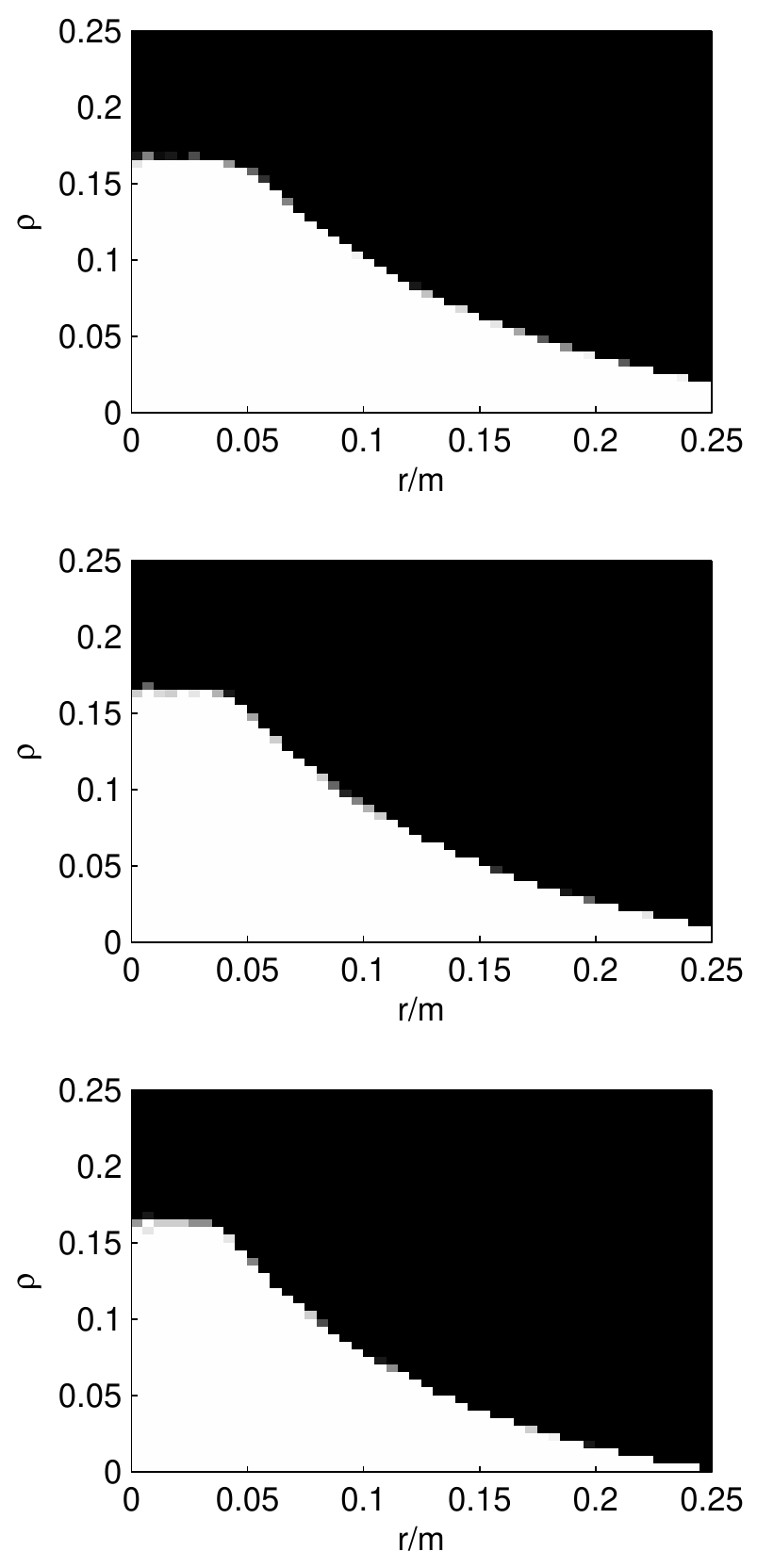}}
\subfloat[]{\includegraphics[width=.55\columnwidth]{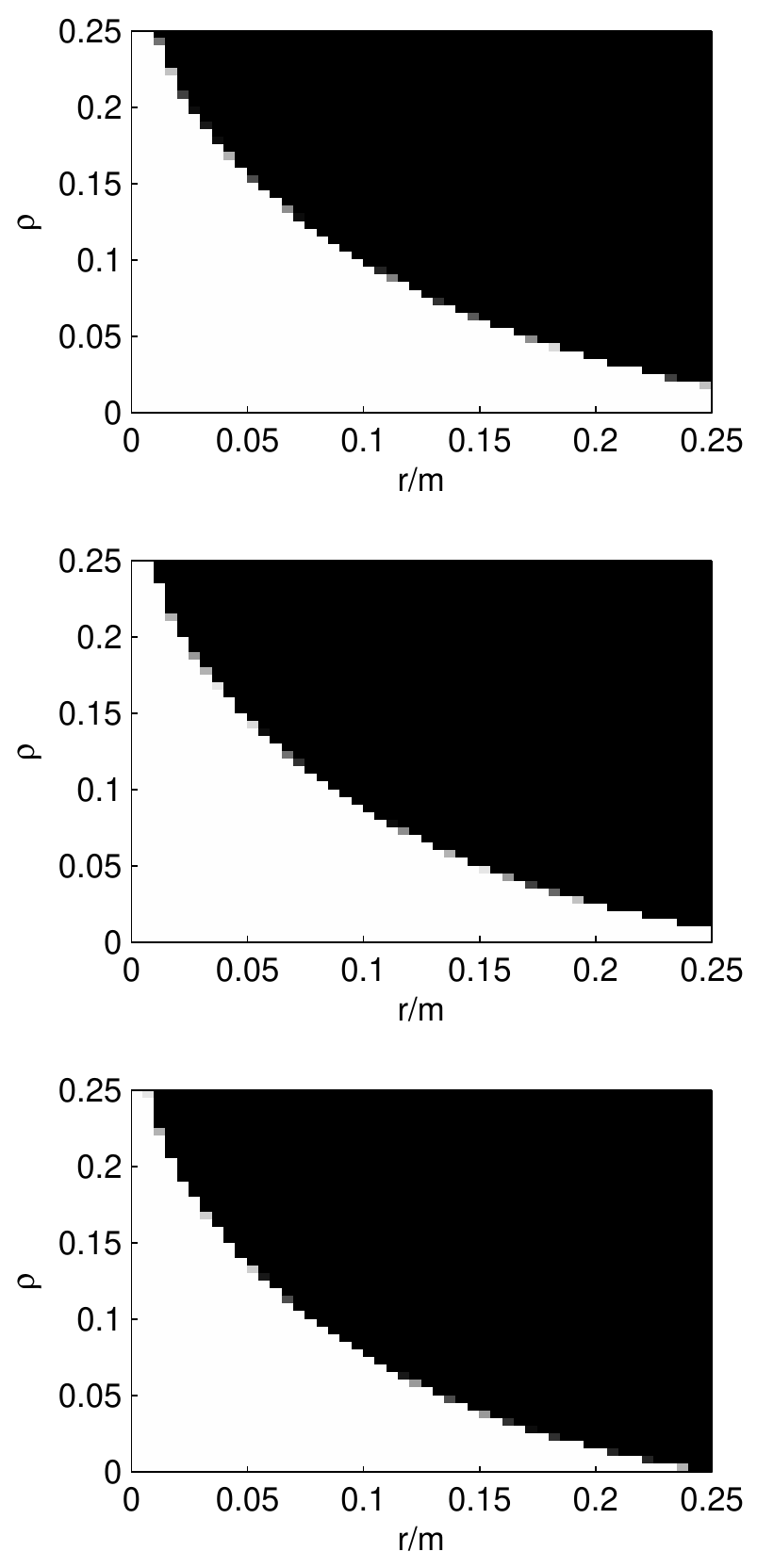}}
\caption{Recovery success rates for (a) polar $4$-complex embedding, (b) polar $2$-bicomplex embedding, and (c) quaternionic embedding. Matrix generation and success criteria are detailed in Section~\ref{sec:numerical}. From top to bottom: results for $\varepsilon=0.1$, $0.05$, and $0.01$, respectively. Grayscale color indicates the fraction of success (white denoting complete success, black denoting total failure).}
\label{fig:phasetrans}
\end{figure*}

\section{Numerical Simulations}
\label{sec:numerical}

To demonstrate the benefit of algebraic closure in polar $n$-bicomplex numbers (introduced in Section~\ref{sec:n-bicomplex}), we will numerically recover hypercomplex matrices of various ranks from additive noises with different levels of sparsity using hypercomplex PCP. Low-rank plus sparse matrices can be generated using Cand{\`e}s et al.'s $\mathbf{XY}^*+\mathbf{S}$ model \cite{Candes11}, where $\mathbf{X}$ and $\mathbf{Y}$ are $m\times r$ matrices with independent and identically distributed (i.i.d.) Gaussian entries from $\mathcal{N}(0,1/m)$, $\mathbf{S}$ is an $m\times m$ matrix with i.i.d. 0-1 entries from $\mathrm{Bernoulli}(\rho)$ multiplied by uniformly random signs, and $r$ and $\rho$ are the desired rank and sparsity, respectively. To accomodate complex coefficients, we instead use the complex normal distribution $\mathcal{CN}(0,\mathbf{I}/m)$ for $\mathbf{X}$ and $\mathbf{Y}$, and replace the random signs for $\mathbf{S}$ with unit-modulus complex numbers whose phases are uniformly distributed. Following \cite{Candes11}, we consider square matrices of size $m=400$. For each $(r,\rho)$ pair, we conduct 10 trials of the following simulation. In each trial, we generate two complex matrices, $\mathbf{M}_1=\mathbf{X}_1\mathbf{Y}_1^*+\mathbf{S}_1$ and $\mathbf{M}_2=\mathbf{X}_2\mathbf{Y}_2^*+\mathbf{S}_2$, using the complexified model described above. Then we embed the two complex matrices into one hypercomplex matrix by:

\begin{enumerate}
\item Polar $4$-complex embedding: the matrices are combined into $(\real\mathbf{M}_1)+(\imag\mathbf{M}_1)e_1+(\real\mathbf{M}_2)e_2+(\imag\mathbf{M}_2)e_3$.
\item Polar $2$-bicomplex embedding: the matrices are combined into $\mathbf{M}_1+\mathbf{M}_2e_1$.
\item Quaternionic embedding \cite{Chan16}: the matrices are combined into $\mathbf{M}_1+\mathbf{M}_2\jmath$.
\end{enumerate}

For each embedding, we perform PCP with a relative error tolerance of $10^{-7}$, as in \cite{Lin09}. We call the $\mathbf{M}_1$ part of the trial a success if the recovered low-rank solution $\mathbf{L}_1$ satisfies $\|\mathbf{L}_1-\mathbf{X}_1\mathbf{Y}_1^*\|_F/\|\mathbf{X}_1\mathbf{Y}_1^*\|_F<\varepsilon$. Likewise, the $\mathbf{M}_2$ part of the trial is deemed successful if the recovered $\mathbf{L}_2$ satisfies $\|\mathbf{L}_2-\mathbf{X}_2\mathbf{Y}_2^*\|_F/\|\mathbf{X}_2\mathbf{Y}_2^*\|_F<\varepsilon$.

The results are shown in Fig.~\ref{fig:phasetrans} for $\varepsilon=0.1$, $0.05$, and $0.01$. The color of each cell indicates the proportion of successful recovery for each $(r,\rho)$ pair across all 10 trials. Results suggest that quaternions and polar $2$-bicomplex numbers have comparable performance up to a sparsity of about $0.16$. Both markedly outperform polar $4$-complex numbers for all $\varepsilon$. As we decrease $\varepsilon$ to $0.01$, the polar $4$-complex numbers have completely failed while the other two are still working well. It may be argued that the quaternions are better than polar $2$-bicomplex numbers for sparsities above $0.16$, but their main weakness is that the dimensionality is fixed at $4$ so they are less flexible than polar $n$-bicomplex numbers in general. In summary, our simulations have provided clear evidence for the importance of algebraic closure in hypercomplex systems.

Next, we will use real data to test the practicality of our proposed algorithms.

\section{Experiments}
\label{sec:exp}

In this section, we use the singing voice separation (SVS) task to evaluate the effectiveness of the polar $n$-bicomplex PCP. SVS is an instance of blind source separation in the field of music signal processing, and its goal is to separate the singing voice component from an audio mixture containing both the singing voice and the instrumental accompaniment (see Fig.~\ref{fig:pcp-svs}). For applications such as singer modeling or lyric alignment \cite{Zhu13}, SVS has been shown an important pre-processing step for better performance. We consider SVS in this evaluation because PCP has been found promising for this particular task, showing that to a certain degree the magnitude spectrogram of pop music can be decomposed into a low-rank instrumental component and a sparse voice component \cite{Huang12}.

\begin{figure}[!t]
\centering
\includegraphics[width=\columnwidth]{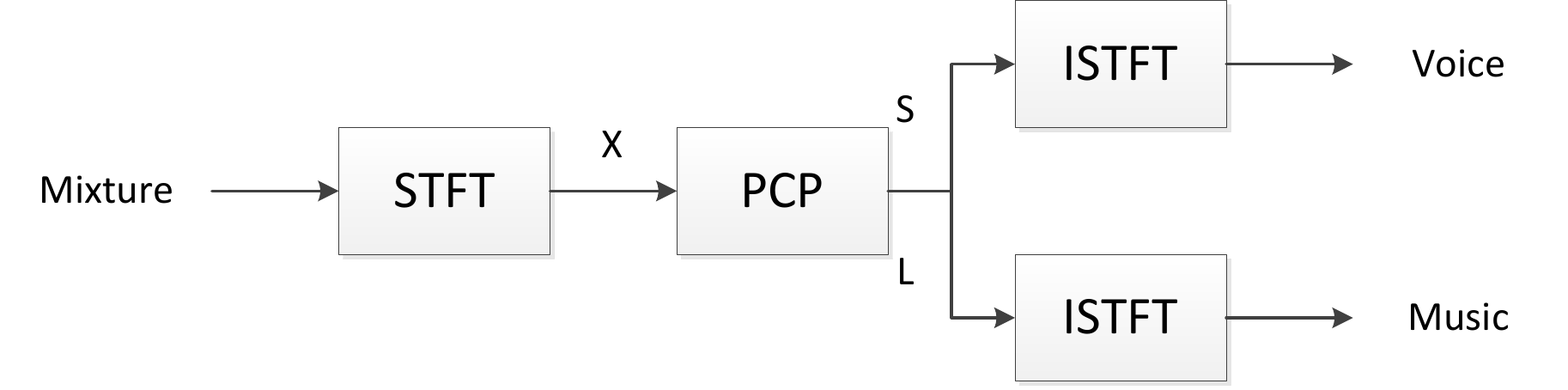}
\caption{Block diagram of a multichannel PCP-SVS system. For our experiments, PCP is either polar $n$-bicomplex PCP, polar $2n$-complex PCP, quaternionic PCP \cite{Chan16}, or tensor RPCA \cite{Zhang14}.}
\label{fig:pcp-svs}
\end{figure}

\subsection{Algorithms}

The following versions of PCP-SVS are compared:
\begin{enumerate}
\item Polar $n$-bicomplex PCP: the $n$-channel audio is represented using $\mathbf{X}_1e_0+\ldots+\mathbf{X}_ne_{n-1}$, where $\mathbf{X}_i$ contains the complex spectrogram for the $i$-th channel.
\item Polar $2n$-complex PCP: the $n$-channel audio is represented using $(\real\mathbf{X}_1)e_0+(\imag\mathbf{X}_1)e_1+\ldots+(\real\mathbf{X}_n)e_{2n-2}+(\imag\mathbf{X}_n)e_{2n-1}$, where $\mathbf{X}_i$ contains the complex spectrogram for the $i$-th channel.
\item Quaternionic PCP (if applicable) \cite{Chan16}: the two-channel audio is represented using $\mathbf{X}_1+\mathbf{X}_2\jmath$, where $\mathbf{X}_i$ contains the complex spectrogram for the $i$-th channel.
\item Tensor RPCA \cite{Zhang14}: the same spectrograms are represented by complex matrices of tubes. The tensor RPCA is used, which is defined by:
\begin{equation}
\label{eq:50}
\min_{\mathcal{L},\mathcal{S}}\|\mathcal{L}\|_{TNN}+\lambda\|\mathcal{S}\|_{1,1,2}\mbox{\quad s.t.\quad}\mathcal{X}=\mathcal{L}+\mathcal{S},
\end{equation}
where $\|\mathcal{L}\|_{TNN}$ is defined as the sum of the singular values of all frontal slices of $\mathcal{\hat{L}}$ (obtained by a Fourier transform along each tube) and $\|\mathcal{S}\|_{1,1,2}$ is defined by $\sum_{i,k}\|\mathcal{S}_{ik:}\|_F$ \cite{Zhang14}. To facilitate comparison with polar $n$-bicomplex PCP, we retrofit \eqref{eq:50} into our framework:
\begin{equation}
\label{eq:51}
\min_{\mathbf{L},\mathbf{S}}\|\mathrm{cft}(\mathbf{L})\|_*+\lambda\|\mathbf{S}\|_1\mbox{\quad s.t.\quad}\mathbf{X}=\mathbf{L}+\mathbf{S},
\end{equation}
where $\mathbf{X}\in\mathbb{K}_n^{l \times m}$ is the input. Our optimized implementation is shown in Algorithm~\ref{alg:2c}, where all calculations are done in the frequency domain.
\end{enumerate}

\begin{algorithm}[t]
\caption{Optimized Tensor RPCA (cf.~\cite{Zhang14})}
\begin{algorithmic}[1]
\label{alg:2c}
\REQUIRE $\mathbf{X}\in F^{l\times m},F\in\{\mathbb{K}_n,\mathbb{CK}_n\}$, $\lambda\in\mathbb{R}$, $\boldsymbol{\mu}\in\mathbb{R}^\infty$
\ENSURE $\mathbf{L}$, $\mathbf{S}$
\STATE Let $\mathbf{\hat{S}}=0$, $\mathbf{Y}=\mathbf{X}/\max\left(\|\mathbf{X}\|_2,\lambda^{-1}\|\mathbf{X}\|_\infty\right)$, $k=1$
\STATE ${\mathbf{\hat{X}}}\leftarrow\rm{fft}(\mathbf{X},n,3)$ \COMMENT{Applies $n$-point DFT to each tube.}
\STATE ${\mathbf{\hat{Y}}}\leftarrow\rm{fft}(\mathbf{Y},n,3)$
\WHILE{not converged}
    \STATE $\mathbf{\hat{Z}}\gets\mathbf{\hat{X}}-\mathbf{\hat{S}}+\mu_k^{-1}\mathbf{\hat{Y}}$
    \FOR{$i=1:n$}
        \STATE $[\mathbf{\hat{U}}_{::i},\mathbf{\hat{\Sigma}},\mathbf{\hat{V}}_{::i}]\leftarrow\mathrm{svd}(\mathbf{\hat{Z}}_{::i})$
        \STATE $\mathbf{\hat{\Sigma}}\gets\mathcal{S}_{1/\mu_k}[\mathbf{\hat{\Sigma}}]$
        \STATE $\mathbf{\hat{L}}_{::i}=\mathbf{\hat{U}}_{::i}\mathbf{\hat{\Sigma}}\mathbf{\hat{V}}_{::i}^*$
    \ENDFOR
    \STATE $\mathbf{\hat{S}}\gets\mathrm{prox}_{\lambda\sqrt{n}/\mu_k\|\cdot\|_1}^F(\mathbf{\hat{X}}-\mathbf{\hat{L}}+\mu_k^{-1}\mathbf{\hat{Y}})$
    \STATE $\mathbf{\hat{Y}}\gets\mathbf{\hat{Y}}+\mu_k(\mathbf{\hat{X}}-\mathbf{\hat{L}}-\mathbf{\hat{S}})$
    \STATE $k\gets k+1$
\ENDWHILE
\STATE $\mathbf{L}\leftarrow\mathrm{ifft}(\mathbf{\hat{L}},n,3)$
\STATE $\mathbf{S}\leftarrow\mathrm{ifft}(\mathbf{\hat{S}},n,3)$
\end{algorithmic}
\end{algorithm}

\subsection{Datasets}

The following datasets will be used:
\begin{enumerate}
\item The MSD100 dataset from the 2015 Signal Separation Evaluation Campaign (SiSEC).\footnote{\url{http://corpus-search.nii.ac.jp/sisec/2015/MUS/MSD100_2.zip}} The dataset is composed of 100 full stereo songs of different styles and includes the synthesized mixtures and the original sources of voice and instrumental accompaniment. To reduce computations, we use only 30-second fragments (1'45'' to 2'15'') clipped from each song, which is the only period where all 100 songs contain vocals. The MSD100 songs are divided into 50 development songs and 50 test songs, but SiSEC requires testing to be done on both sets. We will follow their convention here.
\item The Single- and Multichannel Audio Recordings Database (SMARD).\footnote{\url{http://www.smard.es.aau.dk/}} This dataset contains 48 measurement configurations with 20 audio recordings each \cite{Nielsen14}. SMARD configurations consist of four digits ($ABCD$): $A$ denotes the loudspeaker equipment used, $B$ denotes loudspeaker location, $C$ denotes microphone type, and $D$ denotes microphone array locations. To simulate real life recordings, we require that voice and music come from different point sources, that is $B=0$ for voice and $1$ for music or vice versa. Secondly, we require $C=2$ for circular microphone arrays, because they are better for spatial surround audio recording. Further we choose the first circular array which is closest to the sources, which gives us six audio channels. Finally, we require voice and music to have the same $A$ and $D$ so it makes sense to mix the signals. For each chosen configuration, we mix the first 30 seconds of soprano with the first 30 seconds of each of the music signals (clarinet, trumpet, xylophone, ABBA, bass flute, guitar, violin) at 0 dB signal-to-noise ratio. For soprano, we pad zero until it reaches 30 seconds; for music, we loop it until it reaches 30 seconds. This creates a repeating music accompaniment mixed with sparser vocals. We single out two configurations as the training set (music from 2020 with soprano from 2120, music from 2021 with soprano from 2121), while using the remaining 10 configurations for testing.

\end{enumerate}
For both datasets, we downsample the songs to 22\,050 Hz to reduce memory usage, then we use a short-time Fourier transform (STFT) with a 1\,411-point Hann window with 75\% overlap as in \cite{Chan15}. 

\subsection{Parameters and Evaluation}

Following \cite{Lin09}, the convergence criteria is $\|\mathbf{X}-\mathbf{L}_k-\mathbf{S}_k\|_F/\|\mathbf{X}\|_F<10^{-7}$, and $\boldsymbol{\mu}$ is defined by $\mu_0=1.25/\|\mathbf{X}\|_2$ and $\mu_{k+1}=1.5\mu_k$. The value of $c$ is determined by a grid search on the training set and is found to be 3 for SiSEC and 2 for SMARD (1 for SMARD with tensor RPCA).

The quality of separation will be assessed by BSS Eval toolbox version 3.0\footnote{\url{http://bass-db.gforge.inria.fr/}} in terms of signal-to-distortion ratio (SDR), source-image-to-spatial-distortion ratio (ISR), source-to-interference ratio (SIR), and sources-to-artifacts ratio (SAR), for the vocal and the instrumental parts, respectively \cite{Vincent12}. BSS Eval decomposes each estimated source $h$ into four components (assuming that the admissible distortion is a time-invariant filter \cite{Vincent06}):
\begin{equation}
\hat{s}_h=s_h^\mathrm{true}+e_h^\mathrm{spat}+e_h^\mathrm{interf}+e_h^\mathrm{artif},
\end{equation}
where $\hat{s}$ is the estimated source, $s^\mathrm{true}$ is the true source, $e^\mathrm{spat}$ is the spatial distortion for multi-channel signals, $e^\mathrm{interf}$ is the interference from other sources, and $e^\mathrm{artif}$ is the artifacts of the source separation algorithm such as musical noise. The metrics are then computed as follows \cite{Vincent12}:
\begin{gather}
\mathrm{SDR}_h=20\log_{10}\frac{\|s_h^\mathrm{true}\|}{\|\hat{s}_h-s_h^\mathrm{true}\|},\\
\mathrm{ISR}_h=20\log_{10}\frac{\|s_h^\mathrm{true}\|}{\|e_h^\mathrm{spat}\|},\\
\mathrm{SIR}_h=20\log_{10}\frac{\|s_h^\mathrm{true}+e_h^\mathrm{spat}\|}{\|e_h^\mathrm{interf}\|},\\
\mathrm{SAR}_h=20\log_{10}\frac{\|\hat{s}_h-e_h^\mathrm{artif}\|}{\|e_h^\mathrm{artif}\|}.
\end{gather}
All these measures are energy ratios expressed in decibels. Higher values indicate better separation quality. During parameter tuning, $h$ is dropped and the measures are averaged over all sources. From SDR we also calculate the normalized SDR (NSDR) by computing the improvement in SDR using the mixture itself as the baseline \cite{Hsu10}. We compute these measures for each song and then report the average result (denoted by the G prefix) for both the instrumental (L) and vocal (S) parts. The most important metric is GNSDR which measures the overall improvement in source separation performance.

\subsection{Results}

The results for the MSD100 dataset are shown in Table~\ref{tab:2}. The best results are highlighted in bold. Broadly speaking, polar $2$-bicomplex PCP has the highest GNSDR in both L and S, followed by polar $4$-complex PCP. Both are also slightly better than tensor RPCA on all other performance measures except GISR and GSAR in L. Overall, the result for L is better than S because the instruments in this dataset are usually louder than the vocals (as reflected by the GSDR). It can be observed that the GNSDR for polar $n$-(bi)complex PCP are not inferior to that of quaternionic PCP, suggesting that they are good candidates for PCP with four-dimensional signals.

For the SMARD dataset, the results are presented in Table~\ref{tab:2b}. Both of our proposed algorithms are equally competitive, and both clearly outperform tensor RPCA in terms of GNSDR, GSDR, and GSIR. When we break down the results by configuration, we find that polar $n$-(bi)complex PCP are better than tensor RPCA in 8 out of 10 configurations.

\begin{table}[!t]
\renewcommand{\arraystretch}{1.1}
\caption{Results for MSD100 instrumental (L) and vocal (S), in dB}
\label{tab:2}
\centering
\begin{tabular}{|c|l|c|c|c|c|c|}
\hline
\multicolumn{2}{|l|}{} & GNSDR & GSDR & GISR & GSIR & GSAR\\
\hline\hline
Polar $2$-bi- & L & \textbf{5.01} & \textbf{10.36} & 19.22 & 10.68 & 23.57\\
complex PCP & S & \textbf{3.20} & \textbf {-1.33} & 2.63 & \textbf{9.02} & 0.44\\
\hline
Polar $4$- & L & 5.00 & 10.35 & 19.19 & 10.67 & 23.59\\
complex PCP & S & 3.18 & -1.35 & 2.62 & 9.00 & 0.43\\
\hline
Quaternionic & L & 5.00 & 10.35 & 18.91 & \textbf{10.71} & 23.25\\
PCP & S & 3.15 & -1.38 & \textbf{2.75} & 8.32 & \textbf{0.57}\\
\hline
\multirow{2}{*}{Tensor RPCA} & L & 4.78 & 10.12 & \textbf{22.80} & 10.13 & \textbf{26.03}\\
& S & 2.91 & -1.62 & 1.32 & 8.53 & -0.64\\
\hline
\end{tabular}
\end{table}

\begin{table}[!t]
\renewcommand{\arraystretch}{1.1}
\caption{Results for SMARD instrumental (L) and vocal (S), in dB}
\label{tab:2b}
\centering
\begin{tabular}{|c|l|c|c|c|c|c|}
\hline
\multicolumn{2}{|l|}{} & GNSDR & GSDR & GISR & GSIR & GSAR\\
\hline\hline
Polar $6$-bi- & L & 2.20 & 5.35 & \textbf{11.53} & \textbf{7.63} & 15.47\\
complex PCP & S & \textbf{2.37} & 2.83 & 5.82 & 7.31 & 12.49\\
\hline
Polar $12$- & L & \textbf{2.21} & \textbf{5.36} & 11.51 & \textbf{7.63} & 15.46\\
complex PCP & S & \textbf{2.37} & \textbf{2.84} & 5.82 & \textbf{7.32} & 12.50\\
\hline
\multirow{2}{*}{Tensor RPCA} & L & 1.42 & 4.57 & 9.55 & 6.83 & \textbf{16.65}\\
& S & 1.58 & 2.05 & \textbf{5.85} & 3.06 & \textbf{14.20}\\
\hline
\end{tabular}
\end{table}

\section{Discussion and Conclusion}
\label{sec:conc}

We believe that we have demonstrated the superiority of our proposed hypercomplex algorithms. Theoretically, the tensor RPCA \cite{Zhang14} is computing the nuclear norm in the CFT space \eqref{eq:51}, which is probably due to an erroneous belief that the CFT is unitary and thus does not change anything \cite{Semerci14}. However, as t-SVD is based on the circulant algebra, where the singular values are also circulants, the two trace norms are not equivalent. As a result, we should not have omitted the ICFT, as tensor RPCA does. This omission is difficult to detect because tensors themselves do not have enough algebraic structures to guide us. In contrast, our formulation includes both the CFT and ICFT steps while computing the SVD of a polar $n$-bicomplex matrix, as described in the paragraph after \eqref{eq:28}, which does not violate the underlying circulant algebra. This observation hints at a new role for hypercomplex algebras---to provide additional algebraic structures that serve as a new foundation for tensor factorization. By way of example, let us consider Olariu's other work, the planar $n$-complex numbers, which have a skew-circulant representation \cite{Olariu02}. As skew circulants are diagonalizable by the skew DFT,\footnote{The skew DFT of $[a_0,a_1,\ldots,a_{n-1}]^T$ is $[A_0,A_1,\ldots,A_{n-1}]^T$ where $A_k=\sum_{i=0}^{n-1}a_ie^{-\pi ij(2k+1)/n}$ for $k=0,1,\ldots,n-1$ \cite{Good86}.} a new kind of t-SVD can be derived easily (see Algorithm \ref{alg:3}). Here $\mathrm{sft}$ and $\mathrm{isft}$ stands for skew DFT and inverse skew DFT, respectively.

\begin{algorithm}
\caption{t-SVD with a Skew-Circulant Representation}
\begin{algorithmic}[1]
\label{alg:3}
\REQUIRE $\mathcal{X}\in\mathbb{C}^{l\times m\times n}$
\ENSURE $\mathcal{U}$, $\mathcal{S}$, $\mathcal{V}$
\STATE ${\mathcal{\hat{X}}}\leftarrow\rm{sft}(\mathcal{X},n,3)$\qquad\COMMENT{We use the skew DFT instead.}
\FOR{$i=1:n$}
    \STATE $[\mathbf{\hat{U}}_{::i},\mathbf{\hat{S}}_{::i},\mathbf{\hat{V}}_{::i}]\leftarrow\mathrm{svd}(\mathbf{\hat{X}}_{::i})$
\ENDFOR
\STATE $\mathcal{U}\leftarrow\mathrm{isft}(\mathcal{\hat{U}},n,3);\ \mathcal{S}\leftarrow\mathrm{isft}(\mathcal{\hat{S}},n,3);\ \mathcal{V}\leftarrow\mathrm{isft}(\mathcal{\hat{V}},n,3)$
\end{algorithmic}
\end{algorithm}

What is more, the above procedure can be trivially extended to any commutative group algebras,\footnote{Hypercomplex algebras where the real and imaginary units obey the commutative group axioms including associativity, commutativity, identity, and invertibility.} since the matrix representation of a commutative group algebra is diagonalizable by the DFT matrix for the algebra \cite{Clausen93}, viz.~$\mathbf{F}_{n_1}\otimes\cdots\otimes\mathbf{F}_{n_m}$ where $n_1$ to $n_{m}$ can be uniquely determined \cite{Apple70}. In other words, we get the commutative group algebraic t-SVD simply by reinterpreting $\mathrm{fft}$ and $\mathrm{ifft}$ in Algorithm~\ref{alg:1} according to the algebra's DFT matrix, for which fast algorithms are available \cite{Apple70}. Going even further, we conjecture that the most fruitful results for hypercomplex SVD may originate from regular semigroup algebras (i.e., by relaxing the group axioms of identity and invertibility to that of pseudoinvertibility \cite{Kilp00}). By doing so, we gain a much larger modeling space (see Table~\ref{tab:4}) which may be desirable for data fitting applications. At present, harmonic analysis on semigroups \cite{Berg84} is still relatively unexplored in tensor signal processing.

Regarding the hyperbolic numbers and tessarines that Alfsmann has recommended, we find that both of them share the same circulant representation \cite{Alfsmann06}:
\begin{equation}
\begin{bmatrix}
a_0 & a_1\\
a_1 &a_0
\end{bmatrix},
\end{equation}
where $a_0,a_1\in\mathbb{R}$ for hyperbolic numbers and $a_o,a_1\in\mathbb{C}$ for the tessarines. Thus, the hyperbolic numbers are isomorphic to $\mathbb{K}_2$ whereas the tessarines are isomorphic to $\mathbb{CK}_2$. Interestingly, the seminal paper on tessarine SVD \cite{Pei08} has advocated the $e_1-e_2$ form to simplify computations, where they transform the inputs with $(a_0,a_1)\mapsto(a_0+a_1,a_0-a_1)$, perform the SVDs, then transform the outputs back with $(A_0,A_0)\mapsto((A_0+A_1)/2,(A_0-A_1)/2)$. If we look closely, these are actually Fourier transform pairs (as used in Algorithm~\ref{alg:1}), hence the tessarine SVD can be considered as a special case of t-SVD when $n=2$. It can also be observed that, when $n=1$, the polar $n$-complex and polar $n$-bicomplex PCP degenerate into the real and complex PCP, respectively. It should be emphasized that the complex numbers are not in $\mathbb{K}_n$, therefore we have introduced $\mathbb{CK}_n$ for algebraic closure, and its importance has been confirmed by numerical simulations. We further note that the two families of $2^N$-dimensional hypercomplex numbers introduced by Alfsmann \cite{Alfsmann06} are also commutative group algebras diagonalizable by the Walsh-Hadamard transform matrices $\mathbf{F}_2\otimes\cdots\otimes\mathbf{F}_2$ \cite{Clausen93,Alfsmann06}.

To conclude, we have extended the PCP to the polar $n$-complex and $n$-bicomplex algebras, with good results. Both algebras are representationally compact (does not require $2^N$ dimensions) and are computationally efficient in Fourier space. We have found it beneficial to incorporate hypercomplex algebraic structures while defining the trace norm. More concretely, we have proven an extended von Neumann theorem, together with an adaptation of the group lasso, which in concert enable us to formulate and solve the hypercomplex PCP problem. In doing so, we are able to incorporate the correct algebraic structures into the objective function itself. We have demonstrated that the hypercomplex approach is useful because it can: 1) inform t-SVD-related algorithms by imposing relevant algebraic structures; and 2) generate new families of t-SVD's beyond the circulant algebra. We have also established that tessarine SVD is a special case of t-SVD, and that the $2^N$-hypercomplex family of Alfsmann is amenable to a straightforward extension of t-SVD which we call the commutative group algebraic t-SVD. Having formulated the first proper PCP algorithm on cyclic algebras, we would recommend more crossover attempts between the hypercomplex and tensor-based approaches for future work.

\begin{table}[!t]
\caption{Number of Distinct (Semi)groups of Orders up to 9. From the On-Line Encyclopedia of Integer Sequences, http://oeis.org/A000688 and http://oeis.org/A001427}
\label{tab:4}
\centering
\begin{tabular}{|c|c|r|}
\hline
Order & Number of Com- & Number of Reg-\\
$n$ & mutative Groups & ular Semigroups\\
\hline\hline
1 & 1 & 1\\
2 & 1 & 3\\
3 & 1 & 9\\
4 & 2 & 42\\
5 & 1 & 206\\
6 & 1 & 1\,352\\
7 & 1 & 10\,168\\
8 & 3 & 91\,073\\
9 & 2 & 925\,044\\
\hline
\end{tabular}
\end{table}

\section*{Acknowledgment}

The authors would like to thank the anonymous reviewers for their numerous helpful suggestions.

\bibliographystyle{IEEEtran}
\bibliography{chan16tsp}

\begin{IEEEbiography}[{\includegraphics[width=1in,height=1.25in,clip,keepaspectratio]{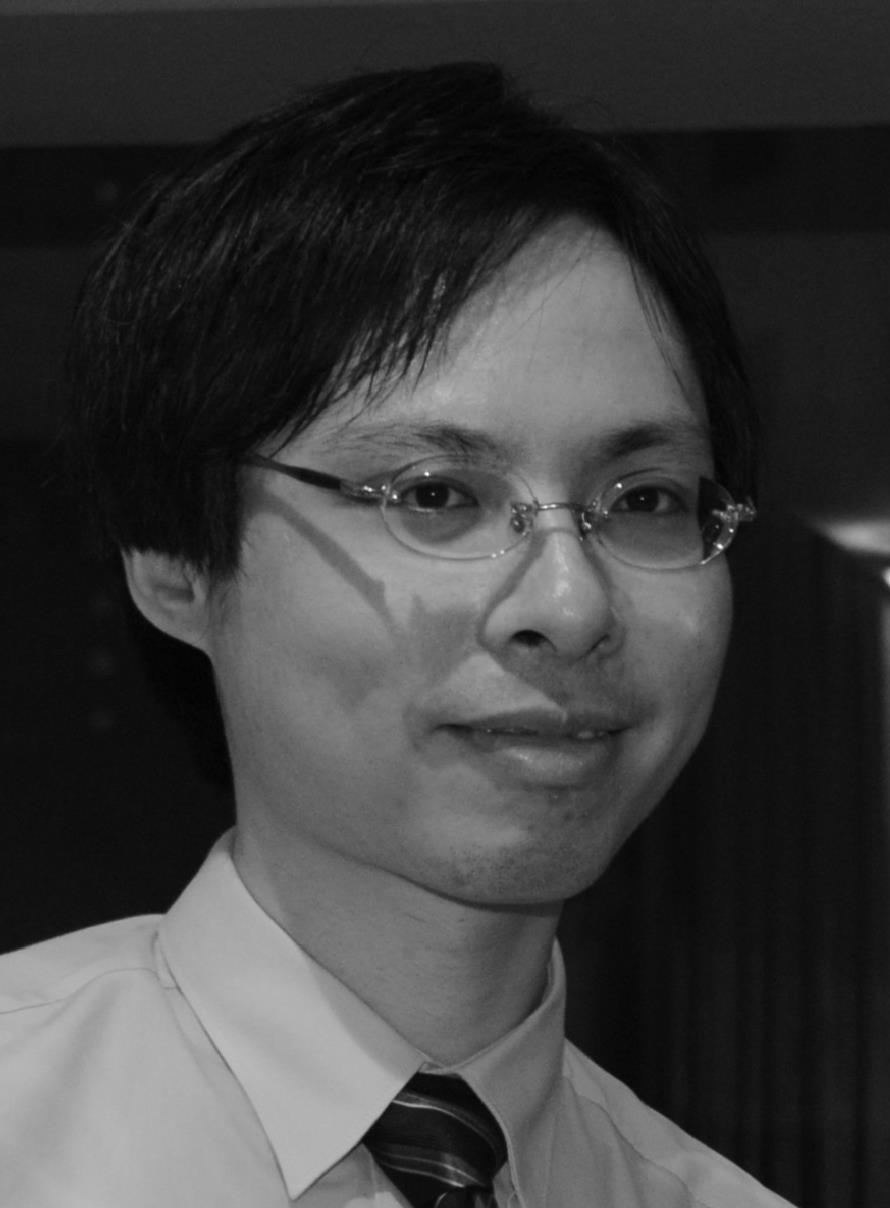}}]{Tak-Shing~T.~Chan}
(M'15) received the Ph.D. degree from the University of London in 2008. From 2006 to 2008, he was a Scientific Programmer at the University of Sheffield. In 2011, he worked as a Research Associate at the Hong Kong Polytechnic University. He is currently a Postdoctoral Fellow at the Academia Sinica, Taiwan. His research interests include signal processing, cognitive informatics, distributed computing, pattern recognition, and hypercomplex analysis.
\end{IEEEbiography}

\begin{IEEEbiography}[{\includegraphics[width=1in,height=1.25in,clip,keepaspectratio]{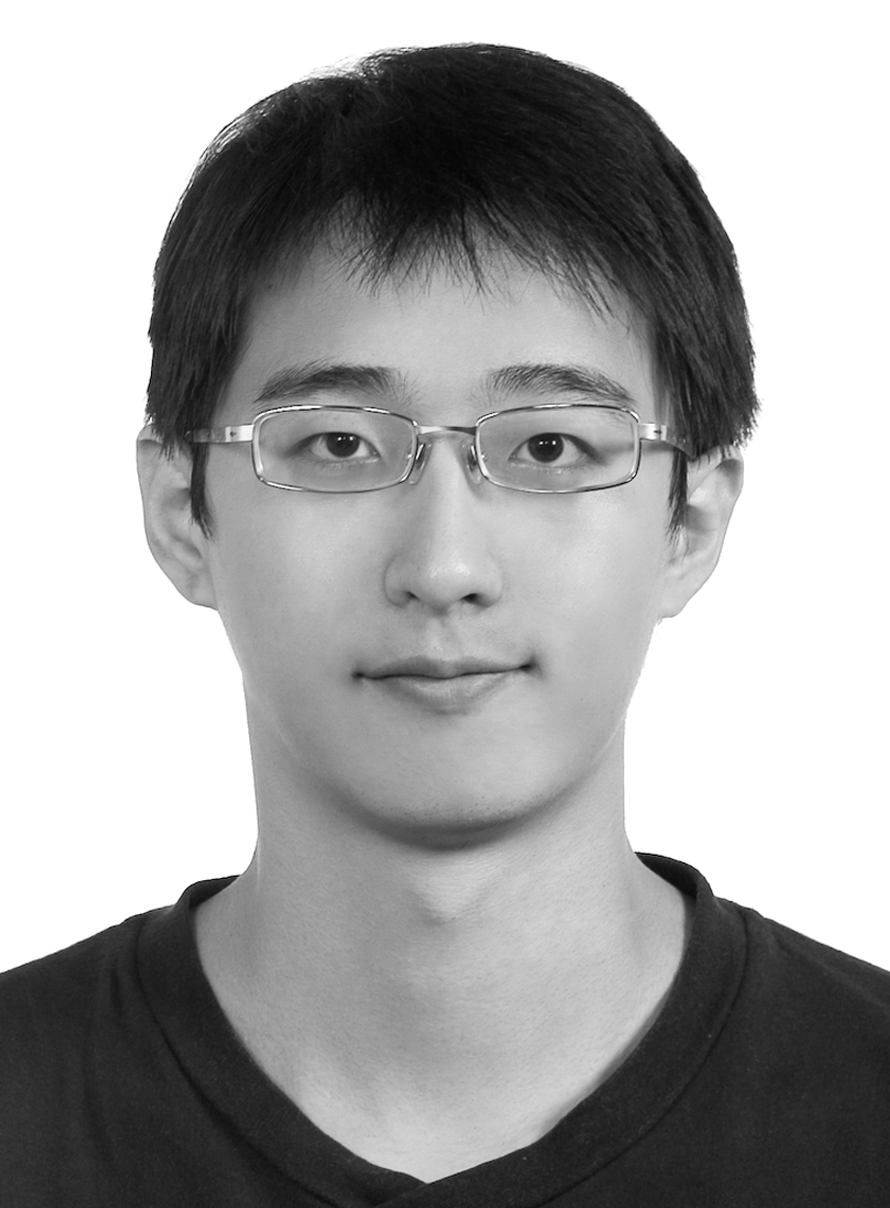}}]{Yi-Hsuan~Yang}
(M'11) received the Ph.D. degree in communication engineering from National Taiwan University in 2010. Since 2011, he has been affiliated with Academia Sinica as an Assistant Research Fellow. He is also an Adjunct Assistant Professor with the National Cheng Kung University. His research interests include music information retrieval, machine learning and affective computing. Dr. Yang was a recipient of the 2011 IEEE Signal Processing Society (SPS) Young Author Best Paper Award, the 2012 ACM Multimedia Grand Challenge First Prize, the 2014 Ta-You Wu Memorial Research Award of the Ministry of Science and Technology, Taiwan, and the 2014 IEEE ICME Best Paper Award. He is an author of the book \emph{Music Emotion Recognition} (CRC Press 2011) and a tutorial speaker on music affect recognition in the International Society for Music Information Retrieval Conference (ISMIR 2012). In 2014, he serve as a Technical Program Co-Chair of ISMIR, and a Guest Editor of the \textsc{IEEE Transactions on Affective Computing}, and the ACM Transactions on Intelligent Systems and Technology.
\end{IEEEbiography}

\end{document}